\documentclass{llncs}
\usepackage{graphicx}
\usepackage{latexsym}
\usepackage{amsmath}
\newcommand{\hide}[1]{}

\newcommand{\Nin}{{N^{-}}}

\newcommand{\res}{{\mathop{\rm\bf res}}}

\newcommand{\dang}{{\mathop{\rm\bf dang}}}
\newcommand{\prop}{{\mathop{\rm\bf prop}}}
\newcommand{\rej}{{\mathop{\rm\bf rej}}}
\newcommand{\tent}{{\mathop{\rm\bf tent}}}
\newcommand{\pend}{{\mathop{\rm\bf pend}}}

\newcommand{\tgs}{{\mathop{\rm\bf tgs}}}
\newcommand{\attraction}{{\mathop{\rm\bf attraction}}}
\newcommand{\ousted}{{\mathop{\rm\bf ousted}}}

\begin{document}
\title{Finalizing Tentative Matches from Truncated Preference
Lists}
\author{Hisao Tamaki}
\authorrunning{H.~Tamaki}
\institute{Department of Computer Science, Meiji University\\
\email{tamaki@cs.meiji.ac.jp}}

\maketitle

\begin{abstract}
Consider the standard hospitals/residents problem, or the
two-sided many-to-one stable matching problem, and assume
that the true preference lists of both sides are
complete (containing all members of the opposite side) and 
strict (having no ties).  The lists actually submitted, however,
are truncated.  Let $I$ be such a truncated instance.
When we apply the resident-proposing deferred acceptance algorithm
of Gale and Shapley to $I$, the algorithm produces a set of tentative
matches (resident-hospital pairs). We say that a tentative match in this set
is {\em finalizable in $I$} if it is in the resident-optimal
stable matching for every completion of $I$ (a complete instance of which $I$ 
is a truncation). We study the problem we call FTM (Finalizability of Tentative
Matches) of deciding if a given tentative match is finalizable in a given
truncated instance. We first show that FTM is coNP-complete, even in the stable
marriage case where the quota of each hospital is restricted to be 1. 

We then introduce and study a special case:
we say that a truncated instance is {\em resident-minimal}, if further
truncation of the preference lists of the residents inevitably changes the set
of tentative matches. Resident-minimal instances are not only practically
motivated but also useful in computations for the general case. We give a computationally useful characterization of negative instances of 
FTM in this special case,  which, for instance,
can be used to formulate an integer program for FTM. 
For the stable marriage case, in particular, 
this characterization yields a polynomial time algorithm
to solve FTM for resident-minimal instances.
On the other hand, we show that FTM remains coNP-complete for resident-minimal instances, 
if the maximum quota of the hospitals is 2 or larger.

We also give a polynomial-time decidable sufficient condition for
a tentative match to be finalizable in the general case.
Simulations show that this sufficient condition
is extremely useful in a two-round matching procedure
based on FTM for a certain type of student-supervisor markets. 
\end{abstract}

\section{Introduction}
In many matching markets in practice which are
modeled by the two-sided, many-to-one 
stable matching problem of Gale and Shapley \cite{gale1962college}@(often
called the {\em hospitals/residents problem}), the task of the participants to 
form a preference list is far from trivial. The large number of participants
in the opposite side makes it practically impossible to evaluate all of them
in enough details to precisely determine the preference order. In some cases,
quite costly procedures such as interviews are involved in the evaluation
process, which makes it even harder to form a precise preference list. 
If each participant is required to submit a complete preference
list in such a market, then the submitted list would be inevitably inaccurate.
Otherwise, the lists submitted would be short.

In the latter case, non-negligible number of participants remain unmatched after
the matching procedure. This is the case, for example, in the National Resident
Matching Program (NRMP) of the United States \cite{NMRP2014}, which assigns candidates 
for residency to positions in hospitals. To provide further opportunities for 
the candidates and positions that have failed to be matched
in the main matching round called MRM (Main Residency Match), 
NRMP organizes a post-match program called SOAP (Supplemental Offer and Acceptance Program).
Unfortunately, the design of this post-match program lacks a
theoretical basis. One clear drawback is that the final matching which 
results from the entire process, MRM followed by SOAP, is not necessarily stable
even under the assumption that, for each agent,
the preference list submitted for MRM followed by the list
used for SOAP, if any, is a truncation of the true preference list, 
despite the popularity of NRMP as a working example of the stable matching model.
Indeed, this drawback does not depend on how SOAP is administered but
is simply due to the possibility that a candidate $r$ unmatched in MRM
may have lost the chance of being accepted by some hospital $h$ he would list
in SOAP because this hospital $h$ has been filled in MRM by candidates possibly
less preferred by $h$ to $r$.  

The algorithmic question studied in this paper is motivated by an approach to
address the above issues: a multi-round stable matching procedure.
This procedure is designed to produce a stable matching as the final outcome,
while allowing participants to incrementally form their preference lists.
In the first round, each participant submits a truncation of its true preference
list, listing only a small number of candidates it ranks the highest. 
The deferred-acceptance (DA)
algorithm of Gale and Shapley \cite{gale1962college} is applied to these
truncated preference lists and stops prematurely with a partial outcome.
In the second and successive rounds, the lists of the participants are
extended, with more agents in the opposite side added in their tails 
as needed to continue the execution of the DA algorithm.

We list potential benefits of such a multi-round procedure.
\begin{enumerate}
  \item In the first round, the participant can concentrate on the
  evaluations of those candidates that are potentially ranked the highest,
  which makes it easier to form an accurate list for the first submission.
  \item In subsequent rounds, some participants do not need to submit extensions
  to their lists, as those extensions are not required by the DA algorithm.
  The saving in the evaluation effort for those participants can be huge.
  \item Even for a participant who does need to submit an extension,
  the evaluation task can be easier since (1)
  some of the remaining candidates may be excluded from considerations as
  it can be deduced (by the central agency) from the submissions so far 
  that they have no possibility to be matched to the participant and  
  (2), except for in the last round,
  the participant may concentrate on those candidate it considers the strongest
  among the remaining candidates, similarly to the situation in the first
  round.
\end{enumerate}

We model the situation after each round in such a multi-round procedure as
follows. The matching is to be made between the set $R$ of residents
and the set $H$ of hospitals. We assume that the true preference lists of 
both the residents and hospitals are
complete, listing all members on the opposite side, and strict, allowing no ties. 
If $J$ is an instance with complete preference lists and
$I$ is obtained from $J$ by truncating some preference lists in $J$,
then we call $I$ a {\em truncation} of $J$ and $J$ a {\em completion}
of $I$. We allow truncations on both sides.

When we run the resident-proposing DA algorithm on an instance $I$ 
executing as many steps as possible in the absence of the
missing parts of the preference lists, the execution results in the
set of {\em tentative matches} for $I$, which we denote by
$\tent(I)$. Here, and throughout the paper, a match means
simply a resident-hospital pair and should not be confused with
a matching, which is a set of matches with certain properties. 
To deal with the truncations of the preference lists
of hospitals to suit our purposes, we need some adaptation of
the standard DA algorithm: hospital $h$ may reject the proposal of
resident $r$ not in its preference list only when $h$ has filled its
quotas by residents which do appear in its list; the proposal of $r$ to $h$ remains
{\em pending} otherwise. See Section~\ref{sec:prelim} for a formal definition
of the adapted DA algorithm and tentative matches for truncated instances.
If $I$ happens to be complete, then $\tent(I)$ is nothing but
the resident-optimal stable matching for $I$ \cite{gale1962college}. When
$I$ is truncated, each match in $\tent(I)$ is truly tentative and may 
eventually be rejected when the DA algorithm continues execution
on some completion of $I$.
We are interested in the following property of tentative matches
and the question on this property.

\begin{definition}
\label{def:finalizable}We say that a match in $\tent(I)$ is {\em
finalizable in $I$} if this match is in $\tent(J)$,
the resident-optimal stable matching for $J$,  
for every completion $J$ of $I$.
\end{definition}

\smallskip
\noindent\textbf{FTM} (Finalizability of Tentative Matches)
\newline\textbf{Instance} An instance $I$ of the hospitals/residents problem
and a match $(r, h) \in \tent(I)$.
\newline\textbf{Question} Is $(r, h)$ finalizable in $I$?
\smallskip

For positive integer $k$, we write $k$-FTM for the version of FTM
where the instances are restricted to those having quota at most $k$
for every hospital: in particular, 1-FTM deals with the one-to-one 
(stable marriage) instances.

In each round of our multi-round procedure, we first compute
$\tent(I)$ by the DA algorithm, where $I$ is the instance specified
by the submissions up to that round, and then compute the set of finalizable matches 
in $\tent(I)$. Those finalizable matches are officially finalized and,
in the succesive rounds, residents in the finalized matches stop
participation and the quotas of the hospitals therein are reduced.
In the final round, each remaining participant is asked to submit 
the complete preference list on the remaining
participants in the opposite side. It is clear from the definition of
finalizability that the matching formed by the entire process is identical to
the one that would be obtained in one round where the participants submit the
complete true preference lists. It is important to note here that 
we do not need to compute the set of finalizable matches exactly: any subset
is sufficient to ensure the correct outcome and a large subset is desirable
for having good progresses through rounds. Thus, even though we have negative
results on the tractability of FTM as described below, they by no means deny the
utility of the notion of finalizability.

Our first result is indeed negative (Theorem~\ref{thm:coNP1}): 
FTM, even 1-FTM in fact, is coNP-complete.

We then look at a special case.  When the resident-proposing
DA algorithm is executed on an instance and
resident $r$ is tentatively matched to hospital $h$ in the outcome,
the tail part of the preference list of $r$ after $h$ 
remains ``unconsumed'' by the algorithm. We say that
an instance $I$ is {\em resident-minimal} if this unconsumed tail
is empty for every resident $r$. Equivalently, 
$I$ is resident-minimal if the set of matches $(r, h)$ such that
$h$ is on the preference list of $r$ equals the set of matches 
that are proposed in the execution of the DA algorithm on $I$.
Resident-minimal instances are of interest for the following reasons.

\begin{enumerate}
  \item A natural and purely algorithm-driven matching procedure
  with incremental submissions would ask for  
  further submissions of participants only when extending their
  preference lists is absolutely necessary for a progress.  In a
  procedure that applies this policy on residents, 
  the instance we have at each execution step
  is resident-minimal. 
  \item 
  Suppose we use a backtrack algorithm to decide if a match is finalizable
  in a general truncated instance $I$, which executes the DA algorithm and
  branches on the next preferred hospital of a resident when it is not given
  in $I$. The extension of $I$ that the algorithm constructs in each branching
  path eventually becomes resident-minimal and the backtrack search
  beyond this search node can be pruned if an efficient algorithm for resident-minimal
  instances is available.
  \item 
  Each general truncated instance $I$ has a further truncation $I'$ that
  is resident-minimal.  The finalizability in $I'$ is a sufficient condition
  for the finalizability in $I$ and therefore an efficient computation
  for resident minimal instances would be useful in estimating 
  the set of finalizable matches in the general case.
\end{enumerate}
Although it is possible to define an analogous notion of hospital-minimal
instances, it is not as natural or as useful as that of resident-minimal
instances mainly because the preference lists of hospitals are not ``sequentially
consumed'' in the resident-proposing algorithm.

We write FTM-RM for FTM (and $k$-FTM-RM for $k$-FTM) in which the instances
are restricted to be resident-minimal. Our main result on FTM-RM is a
computationally useful characterization of negative instances of FTM-RM 
(Theorem~\ref{thm:resident-minimal}).  
This characterization may be used, for example, to formulate an integer program 
or to design a dynamic programming algorithm for FTM-RM. We also give a
polynomial time algorithm for 1-FTM-RM, the stable marriage case, based on
this characterization (Theorem~\ref{thm:1-FTM-RM}).
On the other hand, we show that 2-FTM-RM, and hence FTM-RM, 
remains coNP-complete (Theorem~\ref{thm:coNP2}).

We also develop a polynomial-time decidable sufficient condition for
a tentative match being finalizable (Theorem~\ref{thm:safe-unique}). 
This sufficient condition may be used to compute a subset of finalizable matches
in the proposed multi-round stable matching procedures.
We also show that this condition is necessary for 1-FTM-RM
(Theorem~\ref{thm:necessary}). This gives another proof that 1-FTM-RM is 
polynomial time solvable.

\subsubsection*{Applications}
As in typical Japanese Universities, every student in 
the author's department is required to take a one-year research
project course before graduation, supervised by one of the faculty members.
This situation gives rise to a typical
instance of a many-to-one stable matching problem.

Since 2014, the department has been administering a two-round
stable matching procedure. Supervisors first submit a 
complete rank list of students based on grade scores and 
the results of interviews. 
Then, in the first round of matching, each student lists up to 3 most preferred
supervisors in their rank list.  The DA algorithm is executed
on this truncated instance and the resulting tentative matches are
tested for finalizability using the sufficient condition described in 
Section~\ref{sec:safe}.  Only those students
without finalized matches proceed to the second round, where
they submit a complete preference list of the supervisors who
are not filled by finalized matches. 

This two-round procedure has been working well: a somewhat surprisingly
large number of students are finalized in the first round, resulting in
a huge amount of saving in the evaluation effort and stress on the students'
side.
Due to the lack of publicly disclosable statistics, 
we perform simulations in Section~\ref{sec:simulations} 
to reproduce the phenomenon in a transparent manner.

It would be a challenging research topic to study the
feasibility of replacing 
the current NRMP procedure by a two-round matching procedure
in our approach.  
The advantage of having
a stable matching as a final outcome is attractive
and the success in the smaller scale market described above 
is encouraging. The first step of the feasibility study would be to 
compute the finalizability of the matches produced by the main matching
procedure of NRMP in the past, interpreting the preference lists used in the procedure 
as truncations of the true lists.  This task is challenging, because of the
intractability of FTM and the size of the market.  We note, however, that we may
not necessarily need exact solutions.  Reasonably good estimates on 
the number of finalizable matches may be sufficient for our evaluation purposes.
Theorem~\ref{thm:resident-minimal} (a characterization of negative instances in
the resident-minimal case) and Proposition~\ref{prop:safe} together with
Theorem~\ref{thm:safe-unique} (a polynomial time computable sufficient condition
for finalizability) would be indispensable in computing upper and lower bounds
on that number.

\subsubsection*{Related work}
Truncation of preference lists in the two-sided matching model have been
studied in a different context, namely strategic manipulations.
For example, Roth and Rothblum \cite{roth1999truncation} study the
one-to-one matching case and show that there are instances where a participant
on the proposed side (a hospital in our model) may expect to benefit
significantly by truncating its true preference list, even in the situation
where little information on the preference lists of other 
participants is available.
 
More traditionally, incomplete preference lists arise not as
truncations but as true representations of preferences, where
candidates not on the list are simply meant unacceptable.
There has been active research on the complexity of computing stable matchings
for instances allowing incomplete preference lists and/or ties
(see \cite{iwama2008survey} for a survey).

Rastegari {\it et al.} \cite{rastegari2013two} and
Rastegari {\it et al.} \cite{rastegari2014reasoning} address the 
incompleteness of the preference orders that are inevitable
in large markets in practice. 
They analyze matching markets where the participants
submit their preferences in the form of partial orders that
are consistent with their true total orders.  The
authors of \cite{rastegari2013two} aim at optimizing the number of
interviews needed to sufficiently refine the partial orders 
while the authors of \cite{rastegari2014reasoning}
study the complexity of reasoning about the stable matchings
for the true hidden preference orders using the partial information
available. Our present work may be viewed as dealing with a special case
of their model, where the partial preference order can be represented
in the form of a truncation. 
Although successively extending preference lists is restrictive than
successive refinements of partial orders, it allows the participants
to concentrate on selecting top preferences first and help reduce
the chance of regrets in the submitted lists, compared to the case where
complete lists are required in one shot. The advantage of being restrictive
is that we may have more computationally positive results: none of the positive
results in this paper seem to extend to the general partial order model.

A result analogous to the coNP-completeness of 1-FTM (Theorem~\ref{thm:coNP1})
may be found in their work \cite{rastegari2014reasoning}. 
Restricting themselves to the stable marriage case, 
they consider the problem, among others, of deciding if a given match is a
{\em necessary match}, that is, if it is contained in the resident-optimal
(employer-optimal, in their setting) stable matching for every
completion of the given partial orders into total orders.  
They show that this problem is
coNP-complete. Theorem~\ref{thm:coNP1} in our present paper implies that this
hardness holds for special instances 
where the partial orders are restricted to those representable by truncations.
Their result does not imply our result and, moreover,
neither does their proof, since their reduction uses
partial orders that are not representable by truncations.

The rest of this paper is organized as follows. Section~\ref{sec:prelim}
gives preliminaries of this paper. In Section~\ref{sec:coNP} we prove
the coNP-completeness of 1-FTM. In Section~\ref{sec:r-minimal} we study
resident-minimal instances. In Section~\ref{sec:safe}, we give the sufficient
condition for finalizability. In Section~\ref{sec:simulations}, we
describe our simulation results on the student-supervisor assignment
problem. We conclude the paper in Section~\ref{sec:future},
where we point to some directions for future work.

\section{Preliminaries}
\label{sec:prelim}
Formally, an instance $I$ of our hospitals/residents problem 
is a 5-tuple $(R, H,$
$\{q_h\}_{h \in H}, \{\lambda_r\}_{r \in R}, \{\pi_h\}_{h
\in H})$, where $R$ is the set of {\em residents}, $H$ is the set of
{\em hospitals}, $q_h$ for each $h$ is the quota of $h$, 
$\lambda_r$ for each $r \in R$ is
the preference list of $r$ on $H$, 
and $\pi_h$ for each $h \in H$ is the preference list of $h$ on $R$.
The first three components $R$, $H$, and $\{q_h\}_{h \in H}$ will always be
denoted by these symbols and when we speak of various instances in a context,
these components will be common among those instances: only the preference
lists will vary. 
A preference list on a set $S$ is {\em complete} if it lists all members of $S$.
Our instances in general may have preference lists that are not complete.
An instance is {\em resident-complete} ({\em hospital-complete}, resp.) 
if the preference list of each {\em resident} ({\em hospital}, resp.) is
complete. It is {\em complete} if it is both resident- and hospital-complete.
A {\em match} is a pair in $R \times H$: we say match $(r, h)$
{\em involves} $r$ and $h$.
For each set of matches $M$ and $h \in H$, we define  
$\res\,M = \{r \mid (r, h) \in M \mbox{ for some $h \in H$}\}$
and $\res_h\,M = \{r \mid (r, h) \in M\}$.
A set $M$ of matches is a {\em matching} in $I$ if
$M$ contains at most one match that involves $r$, for each $r \in R$, 
and $|\res_h\,M| \leq q_h$ for each $h \in H$.

We use $+$ operator for concatenation of
sequences and for appending or prepending elements to sequences.
A sequence $\alpha$ is a {\em prefix} of a sequence $\beta$, and 
$\beta$ is an {\em extension} of $\alpha$, 
if $\beta$ can be written as $\alpha$ + $\gamma$ for 
some possibly empty sequence $\gamma$. 
The length of sequence $\alpha$ is denoted by $|\alpha|$. 
Our sequences will never have
duplicate elements and therefore the length of $\alpha$ is precisely the
cardinality of the set of elements in $\alpha$.
 
We say that an instance $I$ is an {\em extension} of an instance $J$,
and that $J$ is a {\em truncation} of $I$,
if the preference list of each resident and each hospital in $I$ is
the extension of that in $J$. An extension $I$ of $J$
is {\em resident-changeless} ({\em hospital-changeless}, resp.) 
if the preference list of each resident (hospital, resp.) is
identical in $I$ and $J$. 

We apply the DA algorithm to a truncated instance $I$ and 
stop when further execution is not possible due to 
the truncations.  We represent the intermediate result
of this execution by two sets of matches:
{\em tentative matches} are those that have been proposed
but not rejected in the algorithm execution;
{\em pending matches} are tentative matches that have not been
rejected because of the incompleteness of the preference list of the
hospital involved.  See below for a more precise definition. 

We formalize the execution of our version of the DA algorithm 
as {\em event sequences}.  
Let $I$ be an instance.
An {\em event} for $I$ is either $(r, h)^+$, the {\em
proposal} of a match $(r, h)$, or $(r, h)^{-}$, the {\em
rejection} of a match $(r, h)$. 
Let $\sigma$ be a sequence of events for $I$ 
(or an {\em event sequence} for $I$, for short).  
We say that $(r, h)$ is {\em proposed (rejected, resp.) in $\sigma$} 
if $(r, h)^+$ ($(r, h)^-$, resp.) appears in $\sigma$.
We denote by $\prop(\sigma)$ ($\rej(\sigma)$, resp.) the
set of matches that are proposed (rejected, resp.) in $\sigma$.
We define $\tent(\sigma) = \prop(\sigma) \setminus \rej(\sigma)$ and
call a match in $\tent(\sigma)$ {\em tentative in $\sigma$}. 
In words, a match is tentative in $\sigma$ if it is proposed but
not rejected in $\sigma$. If a tentative match $(r, h)$ in $\sigma$ is such that
$r$ is not in the preference list of $h$ in the instance $I$ then it is a 
{\em pending match in $\sigma$ with respect to $I$}.  We denote by
$\pend_I(\sigma)$ the set of pending matches in $\sigma$ with respect to $I$.

\begin{remark}
In the standard DA algorithm which regards a missing resident $r$ in the
preference list of hospital $h$ as unacceptable to $h$, there is no notion of
pending matches: the proposal of $r$ to $h$ can be immediately rejected.
In our version, such a proposal should be pending, unless $h$ has filled its
quota with proposals from residents in its list, as we intend to continue the
algorithm when more residents are added to the preference list of $h$.
Also note that the set $\tent(\sigma)$ depends only on the event sequence
$\sigma$ and not on the instance while the set $\pend_I(\sigma)$ depends on
the instance.
\end{remark}

Let $I$ be an instance and $M$ an arbitrary set of matches.  
We say that $(r, h) \in M$ is {\em ousted from $M$ in $I$},
if the preference list of $h$ contains at least $q_h$ residents from
$\res_h\,M$ and either $r$ is missing from this list or
preceded by $q_h$ or more residents from $\res_h\,M$ in this list.
In other words, $(r, h)$ is ousted from $M$ in $I$, if,
no matter how the preference list of $h$ in $I$ is completed,
$r$ is not among the top $q_h$ members of $M$ in the completed list.
We let $\ousted_I(M)$ denote the set of matches ousted from $M$ in $I$.
We say that an event sequence is {\em $I$-feasible} if 
it can be shown so by the inductive procedure below.
\begin{enumerate}
\item An empty sequence is $I$-feasible.
\item Suppose an event sequence $\sigma$ is $I$-feasible.
Then, $\sigma + (r, h)^+$ for each $(r, h) \in R \times H$
is $I$-feasible if $r \not\in \res\,\tent(\sigma)$, $(r, h) \not\in
\prop(\sigma)$, $h$ appears in the preference list of $r \in I$,
and $(r, h') \in \rej(\sigma)$ for every $h' \in H$
that precedes $h$ in the preference list of $r$ in $I$.
On the other hand, 
$\sigma + (r, h)^-$ for each $(r, h) \in R \times H$
is $I$-feasible if $(r, h)$ is in $\ousted_I(\prop(\sigma)) \setminus
\rej(\sigma)$. 
\end{enumerate}
\begin{remark}
\label{rem:ousted}
We have $\ousted(\prop(\sigma)) \setminus \rej(\sigma) =
\ousted(\tent(\sigma))$ for $I$-feasible $\sigma$.
Therefore, the condition for the $I$-feasibility of $\sigma + (r, h)^-$ 
above may be expressed as $(r, h) \in \ousted(\tent(\sigma))$, 
which is more consistent with the 
traditional definition of the DA algorithm.
We use the condition in the present form, since it makes the monotonicity of the
feasibility expressed by the following proposition obvious. 
\end{remark}

For brevity, we refer to $I$-feasible event sequences simply as $I$-feasible
sequences.
\begin{proposition}
\label{prop:order_insensitive}
Let $I$ be an instance and let $\sigma + e$ be an $I$-feasible
sequence where $e$ is a single event. Then, for each $I$-feasible
sequence $\sigma'$ that contains all events in $\sigma$ but not $e$, 
$\sigma' + e$ is $I$-feasible.
\end{proposition}

Given instance $I$, the execution of the DA algorithm on $I$ results
in an arbitrary, due to the non-determinacy of the algorithm, but maximal
$I$-feasible sequence.
The following observation that is well known for the standard
DA algorithm \cite{gale1962college} holds also for our variant.
\begin{proposition}
\label{prop:unique}
Let $I$ be an instance and let $\sigma$
and $\sigma'$ be two maximal $I$-feasible sequences.
Then, $\sigma$ and $\sigma'$ contain the same set of events.
\end{proposition}
\begin{proof}
Suppose $\sigma'$ contains an event that is not in $\sigma$.
Among such events, choose one that appears first in $\sigma'$
and call it $e$. Then $\sigma$ followed by $e$ is
$I$-feasible, due to Proposition~\ref{prop:order_insensitive}, 
contradicting the maximality of $\sigma$. Therefore, $\sigma'$ does
not have any event not in $\sigma$ and vice versa.
\qed
\end{proof}

For an instance $I$ with a maximal $I$-feasible sequence
$\sigma$, we denote $\prop(\sigma)$,  $\rej(\sigma)$,
$\tent(\sigma)$, and $\pend_I(\sigma)$ by
$\prop(I)$,  $\rej(I)$,
$\tent(I)$, and $\pend(I)$. This notation is justified
since these sets do not depend on the choice of $\sigma$ and
determined solely by $I$, by Proposition~\ref{prop:unique}.
We say that instance $I$ {\em proposes} ({\em rejects}, resp.) a match if it is
in $\prop(I)$ ($\rej(I)$, resp.).

We say that an event sequence is {\em feasible}
if it is $I$-feasible for some instance $I$.

\section{Hardness of 1-FTM}
\label{sec:coNP}
To prove hardness results we use the following folklore.
A similar statement appears in the description of SATISFIABILITY problem in
Garey and Johnson \cite{garey2002computers}. We include a proof for
self-containedness.

\begin{proposition}
\label{prop:1-2SAT}
SAT is NP-complete, even when restricted to a clause set
in which each variable appears exactly twice positively and exactly
once negatively. 
\end{proposition}
\begin{proof}
Let a clause set $S$ be given. 
We show below that $S$ can be converted into a clause set 
$S'$, without changing the satisfiability, in which 
each variable appears exactly three times and moreover
at least once positively and at least once negatively.
By replacing some variables by their negations if necessary, $S'$ may further
be converted into a clause set satisfying the condition of the
proposition. 

Suppose variable $x$ occurs $k$ times in $S$. We may assume $k \geq 2$
since otherwise the value of $x$ can be fixed without changing the
satisfiability.
Then, we replace occurrences of $x$ by distinct new variables $x_i$, $1 \leq i
\leq k$, and add clauses $\bar{x_i} \vee x_{i + 1}$ for $1 \leq i \leq k$, where
$x_{k + 1} = x_1$, to force these variables to take the same value.
The clause set $S'$ is obtained from $S$ by doing this for all variables.
\qed 
\end{proof}

\begin{theorem}
\label{thm:coNP1}
1-FTM, and hence FTM, is coNP-complete.
\end{theorem}
\begin{proof}
That FTM is in coNP is trivial. We prove the hardness by
reducing SAT to the complement of 1-FTM.

Let $S$ be an arbitrary set of SAT clauses.  Let $X$ be the
set of variables of $S$, and $C_1$, \ldots, $C_m$ the enumeration
of clauses in $S$.  Relying on Proposition~\ref{prop:1-2SAT},
we assume that each variable occurs positively in exactly 
two clauses and negatively in exactly one clause.

We construct an instance 
$I = (R, H, \{q_h\}_{h \in H},$
$\{\pi_h\}_{h \in H}, 
\{\lambda_r\}_{r \in R})$ as follows. 
$R$ consists of two distinguished residents $r_0$ and $r_1$, together with
distinct residents $r_x^1$, $r_x^2$, $p_x^0$, $p_x^1$, and $p_x^2$ for each $x
\in X$. Fix $x \in X$. Let $C_{j_0}$,  $C_{j_1}$, and $C_{j_2}$ be the three
clauses in which $x$ appears and assume that the occurrence of $x$ in
$C_{j_0}$ is negative. We say
that resident $p_x^i$, $i = 0, 1, 2$, is {\em associated with} the occurrence
of $x$ in $C_{j_i}$.
$H$ consists of a distinguished hospital $h_0$
together with distinct hospitals $h_x^{-1}$, $h_x^{-2}$, $h_x^{1}$, and
$h_x^{2}$ for each $x \in X$ and distinct hospitals $h_j$ for each $1 \leq j \leq m$.
The preference lists of the hospitals are as follows.
The unspecified parts of the lists are immaterial.
\begin{enumerate}
  \item The list of $h_0$ starts with $r_0$ followed by $r_1$.
  \item For $x \in X$ and $i = 1, 2$, the list of $h_x^{-i}$ starts with $r_x^i$
  followed by $p_x^{0}$.
  \item For $x \in X$ and $i = 1, 2$, the list of $h_x^{i}$ starts with $r_x^i$
  followed by $p_x^{i}$.
  \item For $1 \leq j \leq m$, the list of $h_j$ starts with the 
  residents (of the form $p_x^{0}$, $p_x^1$, or $p_x^2$ for some $x$) 
  that are associated with the variable occurrences in $C_j$, in an
  arbitrary order, followed by $r_0$.
\end{enumerate}
The preference lists of the residents are as follows.
Unlike in the description above for hospitals, these lists are truncated
after the specified elements.
\begin{enumerate}  \item The list of $r_0$ is $h_1$, $h_2$, \ldots $h_m$, followed by 
  $h_0$.
  \item The list of $r_1$ consists solely of $h_0$.
  \item For each $x \in X$, the lists of $r_x^{1}$ and $r_x^{2}$ are
  empty.
  \item For each $x \in x$, the list of $p_x^{0}$ starts with
  $h_x^{-1}$ and $h_x^{-2}$ in this order, followed by $h_j$, where
  $C_j$ is the clause that contains the variable occurrence with which
  resident $p_x^{0}$ is associated.
  \item For each $x \in x$ and $i = 1, 2$, the list of $p_x^i$ starts with
  $h_x^{i}$ followed by $h_j$, where
  $C_j$ is the clause that contains the variable occurrence with which
  resident $p_x^{i}$ is associated.
\end{enumerate}

See Figure~\ref{fig:reduction0} for an example.

\begin{figure}[htb]
\begin{center}
\includegraphics[height=2.8in,bb=0 0 922 518]{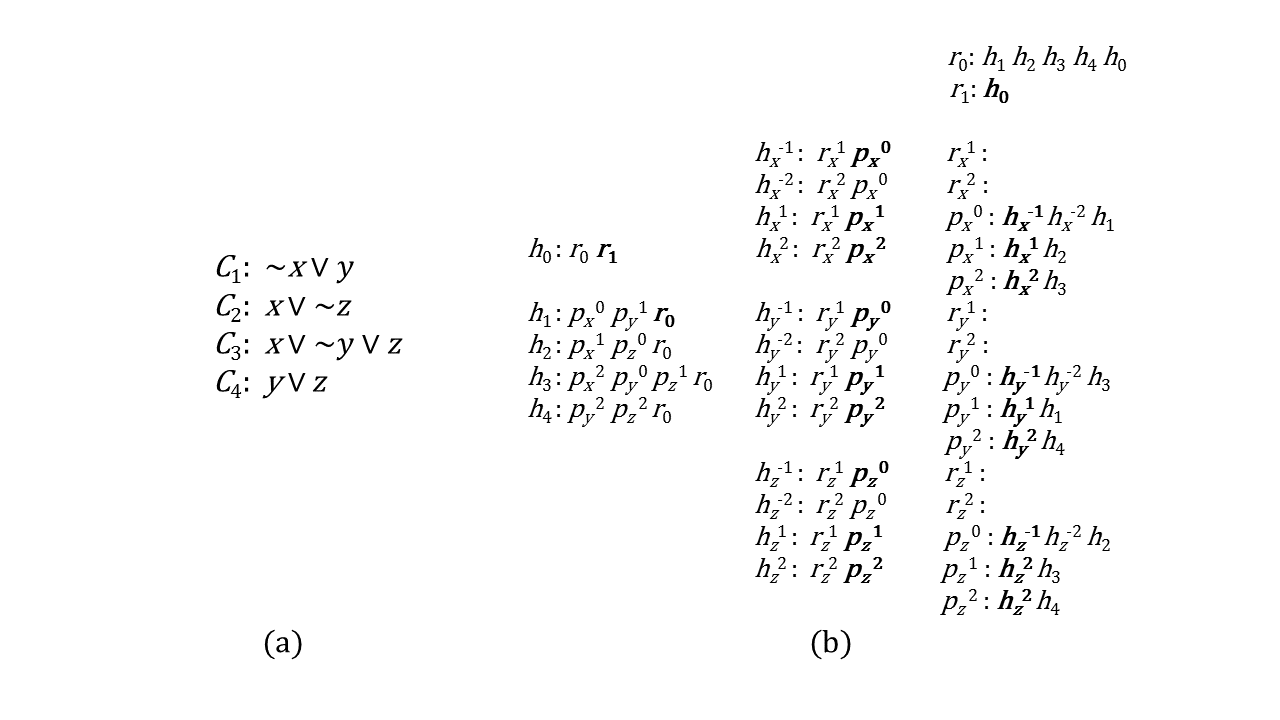}
\end{center}
(a) Clause set $S$\\
(b) The instance of 1-FTM corresponding to $S$\\
Tentative matches are shown in bold face from both sides
\caption{Reduction from SAT to the complement of 1-FTM}
\label{fig:reduction0}
\end{figure}

First observe that $\tent(I) = \{(r_1, h_0), (r_0, h_1)\} \cup 
\{(p_x^{0}, h_x^{-1}), (p_x^{1}, h_x^1), (p_x^{2}, h_x^2) \mid x \in X\}$.
We show that $(r_1, h_0)$ is not finalizable in $I$ if and only
if $S$ is satisfiable.

Let $J$ be an extension of $I$. We say that resident $p$ of the form
$p_x^{i}$ is {\em activated in} $J$, if match $(p, h_j)$ is proposed in $J$,
where $h_j$ is such that $C_j$ contains the 
variable occurrence to which $p$ is associated and hence $h_j$ is the last
entry of the preference list of $p$ in $I$. 
Observe that, $p_x^{0}$ is activated if and only if $r_x^1$ chooses $h_x^{-1}$ and 
$r_x^2$ chooses $h_x^{-2}$ as their first hospitals on their lists.
Similarly, $p_x^{i}$, $i =1, 2$, is activated if and only if 
$r_x^i$ chooses $h_x^{i}$.  Therefore, for each $x \in X$, 
the two events (1) $p_x^{0}$ is activated and (2)
both $p_x^{1}$ and $p_x^{2}$ are activated are mutually exclusive and, moreover,
we may choose the way the lists of $r_x^1$ and $r_x^2$ are extended
so that at least one of (1) and (2) happens.  Thus, the activation
of residents $p_x^0$, $p_x^1$, and $p_x^2$ can properly simulate
the truth assignment to variable $x$.

Also observe that $(r_1, h_0)$ is rejected if and only if 
$(r_0, h_0)$ is proposed, which happens if and only if 
there is a chain of rejections/proposals of resident $r_0$ through the
hospitals $h_1$, \ldots, $h_m$ leading to this proposal.
Since the pair $(r_0, h_j)$ is rejected, provided that this pair is proposed, 
if and only if at least one resident on the list of $h_j$ that
is associated with a variable occurrence in $C_j$ is activated,
we conclude that $S$ is satisfiable if and only if there is
an extension of $I$ in which $(r_1, h_0)$ is rejected.
\qed
\end{proof}

\section{Resident-minimal instances}
\label{sec:r-minimal}
In this section, we study resident-minimal instances.

\subsection{Simple extensions and prescriptions}
In this subsection, we define the notions of simple extensions and
prescriptions, which characterize negative instances of FTM-RM.
Hospital-complete instances play an important role here.  

\begin{proposition}
\label{prop:hospital-completion}
Let $I$ be a resident-minimal instance and $J$ a 
resident-changeless and hospital-complete extension of $I$.
Then $J$ is also resident-minimal.
We also have
$\prop(I) = \prop(J)$ and $\tent(I) \setminus \pend(I) \subseteq \tent(J)
\subseteq \tent(I)$.
\end{proposition}
\begin{proof}
Let $L$ be the set of all matches $(r, h)$ such that $h$ is on the preference
list of $r$ in $I$. 
Since $I$ is resident-minimal, we have $\prop(I) = L$.
As $\prop(I) \subseteq \prop(J) \subseteq L$, we have $\prop(I) = \prop(J)$
and $J$ is resident-minimal.  It immediately follows that $\tent(J)
\subseteq \tent(I)$. Let $(r, h)$ be a match in $\tent(I) \setminus \pend(I)$. 
Then, $r$ is on the preference list of $h$ in $I$ and hence extending the preference list of $h$ does not affect the rank of $r$.
Therefore, we have $(r, h)\in \tent(J)$ and hence
$\tent(I) \setminus \pend(I) \subseteq \tent(J)$.
\qed 
\end{proof}

Let $\sigma$ be an event sequence. We say that an extension $\sigma + \tau$
of $\sigma$ is {\em simple} if $\prop(\tau) \cap \rej(\tau) = \emptyset$ or, in words,
$\tau$ never rejects a proposal made in itself.
We say that an extension $J$ of instance $I$ is {\em simple} if
the maximal $I$-feasible sequence has a simple maximal $J$-feasible extension
or, equivalently, $(\prop(J) \setminus \prop(I)) \cap (\rej(J) \setminus
\rej(I)) = \emptyset$. The goal of this subsection is to show that, for each
resident-minimal instance $I$, we do not need to search through all extensions
of a maximal $I$-feasible sequence to decide the finalizability of a match in
$\tent(I)$: we need only to look at simple extensions. Indeed, it will turn
out that we need only to look at hospital-complete simple extensions.

\begin{proposition}
\label{prop:simple-hcomplete}
Let $I$ be a resident-minimal instance and $J$ a
simple extension of $I$. 
Then, there is a simple and hospital-complete extension $J'$ of $I$
such that $\rej(J) \subseteq \rej(J')$.
\end{proposition}  
\begin{proof}
Let $I$ and $J$ be as in the lemma. We assume without loss of generality that 
$J$ is resident-minimal: if not, take an appropriate truncation.
Let $J_1$ be an arbitrary hospital-complete and resident-changeless extension of
$J$.  Since $J$ is resident-minimal, so is $J_1$ by
Proposition~\ref{prop:hospital-completion}.
We also have $\rej(J) \subseteq
\rej(J_1)$ and, moreover, $(\rej(J_1) \setminus \rej(J)) \subseteq
\pend(J)$, since $J_1$ is a resident-changeless extension of $J$.
We construct a simple extension of $I$ by truncating preference
lists of residents in $J_1$.  
Let $M = (\rej(J_1) \setminus \rej(J)) \setminus \tent(I)$.
For each $(r, h) \in M$, $h$ is the last entry of the preference list of $r$ in
$J$ and hence in $J_1$, since $(r, h) \in \pend(J)$ and $J$ is
resident-minimal.
Let $J'$ be obtained from $J_1$ by, for each match 
$(r, h) \in M$, removing $h$ from the preference list of $r$.
Then, we have $\prop(J') = \prop(J_1) \setminus M$ and
$\rej(J') = \rej(J_1) \setminus M$.  
For each $(r, h) \in M$, $h$ is not on the preference list of
$r$ in $I$, since $(r, h) \not\in (\tent(I) \cup \rej(J)) \supseteq (\tent(I)
\cup \rej(I)) = \prop(I)$.
Therefore, $J'$ is an extension of $I$. 
We claim that it is a simple extension of $I$.  
To see this, observe that $\rej(J') \setminus \rej(J) \subseteq
\tent(I)$ from the construction of $J'$.
Since no match in $(\prop(J') \setminus \prop(I)) \subseteq (\prop(J)
\setminus \prop(I))$ can be in $\rej(J)$ as $J$ is a simple extension of $I$,
no such match can be in $\rej(J')$.  Therefore, $J'$ is a simple
extension of $I$.  As $M \subseteq (\rej(J_1) \setminus
\rej(J))$ and $\rej(J') = \rej(J_1) \setminus M$, 
we have $\rej(J) \subseteq \rej(J')$ and are done.
\qed   
\end{proof}

For the time being, we concentrate on resident-minimal
instances that are also hospital-complete and try to 
characterize their simple extensions.  

\begin{proposition}
\label{prop:presc}
Let $I$ be a resident-minimal and hospital-complete instance and $J$ a simple
extension of $I$. 
Let $P = \prop(J) \setminus \prop(I)$ and $X = \rej(J) \setminus \rej(I)$.
Then, these sets of matches satisfy
the following conditions.
\begin{description}
\item[P1:] $P \cap \prop(I) = \emptyset$.
\item[P2:] For each $r \in R$, there is at most one $h \in H$ such that $(r, h)
\in P$. 
\item[P3:] $X \subseteq \tent(I)$.
\item[P4:] $\res\,P \cap \res\,\tent(I) \subseteq \res\,X$. 
\item[P5:]
  For each $h \in H$, we have $|\res_h(P \cup (\tent(I) \setminus X)))|
\leq q_h$.  Moreover, if $\res_h\,X$ is non-empty then 
  we have $|\res_h(P \cup (\tent(I) \setminus X)))| = q_h$.
\item[P6:]
  For each $h \in H$, 
  each member of $\res_h(P \cup (\tent(I) \setminus X)))$ precedes 
  all members of $\res_h\,X$ in the preference list of $h$ in $I$.
\end{description}
\qed 
\end{proposition}

A {\em prescription} for resident-minimal and hospital complete instance $I$ is
a pair $(P, X)$ of sets of matches that
satisfies the conditions P1 through P6 in Proposition~\ref{prop:presc}.
The {\em target set} of prescription 
$(P, X)$, denoted by $\tgs(P, X)$ is defined by 
$\tgs(P, X) = \{(r, h) \in X \mid r \not\in \res\,P\}$. 
The crucial part of the proof of the result in this section is in showing
that a certain form of converse of Proposition~\ref{prop:presc} holds:
a prescription $(P, X)$ for $I$
implies a simple extension of $I$ that rejects matches in $\tgs(P, X)$.

\begin{proposition}
\label{prop:ready}
Let $I$ be a resident-minimal and hospital-complete instance and
let $(P, X)$ be a prescription for $I$.
Then, we have $|\res\,P \setminus \res\,\tent(I)| \geq |\tgs(P, X)|$.
\end{proposition}
\begin{proof}
Let $Q = \{(r, h) \in P \mid r \not\in \res\,\tent(I)\}$.
By condition P2 for $(P, X)$ being a prescription for $I$,
we have $|Q| = |\res\,P \setminus \res\,\tent(I)|$.
Our goal is to show that $|Q| \geq |\tgs(P, X)|$. 

By condition P5 for $(P, X)$ being a prescription for $I$, we have
$|\res_h\,X| \leq  |\res_h\,P|$ for each $h \in H$. Therefore,
we have $|X| \leq |P|$. 
On the other hand, let $(r, h)$ be a match in 
$(r, h) \in P \setminus Q$.  Because of condition P4, 
there is some $h'$ such that $(r, h') \in X$.  However, by
definition, $(r, h') \not\in \tgs(P, X)$ as $r \in \res\,P$.
Therefore, we have $|P \setminus Q| \leq |X \setminus \tgs(P, X)|$.
As $Q$ is a subset of $P$ and $\tgs(P, X)$ is
a subset of $X$, we have 
$|P| - |Q| \leq |X| - |\tgs(P, X)|$.
Combining this with $|X| \leq |P|$, we conclude that 
$|Q| \geq |\tgs(P, X)|$.
\qed 
\end{proof}

\begin{lemma}
\label{lem:executable}
Let $I$ be a resident-minimal and hospital-complete instance and suppose there
is a  prescription $(P, X)$ for $I$.
Then there is some simple extension $J$ of $I$ such that
$\prop(J) \setminus \prop(I) \subseteq P$, 
$\rej(J) \setminus \rej(I) \subseteq X$, and
$\tgs(P, X) \subseteq \rej(J)$.
\end{lemma}
\begin{proof}
Let $I$ and $(P, X)$ be as in the lemma 
and $\sigma$ a maximal $I$-feasible sequence. 

We prove the statement of the lemma by induction on $|P|$. 
We take $\tgs(P, X) = \emptyset$ as the base case,
which includes the case $P = X = \emptyset$. The statement
is satisfied with $J = I$ in this case.

For the induction step, suppose $\tgs(P, X)$ is non-empty and
let $Q = \{(r, h) \in P
\mid r \not\in \res\,\tent(I)\}$. Since $\tgs(P, X)$ is non-empty,
$Q$ is non-empty by Proposition~\ref{prop:ready}.
Let $\sigma + \tau_1$ be an extension of $\sigma$ such that
$\tau_1$ first lists the proposals of matches in $Q$
in an arbitrary order and then lists all rejections, in an arbitrary
order, that are made possible by these proposals (without further chain of 
proposals and rejections).
Then, $\sigma + \tau_1$ is maximal $I_1$-feasible where
$I_1$ is the extension of $I$ obtained by appending $h$ in 
the preference list of $r$ for each $(r, h) \in Q$.
We claim that $\rej(\tau_1) \subseteq X$. To see this,
fix $h \in H$. By condition P5 for $(P, X)$ being a prescription for
$I$, we have $|\res_h(P \cup (\tent(I) \setminus X)))| \leq q_h$ and
hence $|\res_h(Q \cup (\tent(I) \setminus X)| \leq q_h$.
Thus, no match $(r, h)$ is rejected by $\tau_1$ unless $(r, h) \in X$.
If $\tgs(P, X) \subseteq \rej(I_1)$ then 
we are done with $J = I_1$. 

So suppose otherwise, that $\tgs(P, X)$ is not contained in $\rej(I_1)$.
Consider the prescription $(P_1, X_1)$ for $I_1$ where
$P_1 = P \setminus Q$ and $X_1 = X \setminus \rej(\tau_1)$.
We confirm that this pair is indeed a prescription for $I_1$.
From condition P1 for $(P, X)$ being a prescription for
$I$, we have $P \cap \prop(I) = \emptyset$.
Since $\prop(I_1) = \prop(I) \cup Q$ and $P_1 = P \setminus Q$,
it follows that 
$P_1 \cap \prop(I_1) = \emptyset$, condition P1 for $(P_1, X_1)$
being a prescription for $I_1$.  Condition P2 immediately follows from
the corresponding condition for $(P, X)$.
Since $X \subseteq \tent(I)$ (condition P3 for
$(P, X)$) and $X_1 = X \setminus \rej(\tau_1)$, we have
$X_1 \subseteq \tent(I) \setminus \rej(\tau_1) \subseteq \tent(I_1)$, 
condition P3.

For condition P4, we use the facts that
$P_1$ and $Q$ partition $P$ and that
$Q$ and $\tent(I) \setminus \rej(\tau1)$
partition $\tent(I_1)$. 
Also using condition P4 for $(P, X)$ that 
$\res\,P \cap \res\,\tent(I) \subseteq X$, we have 
\begin{eqnarray*}
\res\,P_1 \cap \res\,\tent(I_1) & = & (\res\,P \setminus \res\,Q)
\cap \res(Q \cup (\tent(I) \setminus \rej(\tau_1)))\\
& = &(\res\,P \setminus \res\,Q)
\cap (\res\,Q \cup (\res\,\tent(I) \setminus \res\,\rej(\tau_1)))\\
& = & \res\,P \cap (\res\,\tent(I) \setminus \res\,\rej(\tau_1))\\
& = & (\res\,P \cap \res\,\tent(I)) \setminus \res\,\rej(\tau_1)\\
& \subseteq & \res\,X \setminus \res\,\rej(\tau_1) \\
& \subseteq & \res(X \setminus \rej(\tau_1)) \\
& = & \res\,X_1.
\end{eqnarray*} 
Therefore, condition P4 holds.
 
For conditions P5 and P6, observe that  
\begin{eqnarray*}
P_1 \cup (\tent(I_1) \setminus X_1) & = &
P_1 \cup ((Q \cup (\tent(I) \setminus \rej(\tau_1))) \setminus X_1)) \\
& = & P_1 \cup (((Q \cup \tent(I)) \setminus \rej(\tau_1))) \setminus X_1)) \\
& = & P_1 \cup (Q \cup (\tent(I) \setminus X)) \\
& = &  P \cup (\tent(I) \setminus X),
\end{eqnarray*}
where we have repeatedly used the disjointness between subsets of $P$ and
subsets of $X$.
Therefore, for each $h \in H$, we have $\res (P_1 \cup (\tent(I_1) \setminus
X_1)) = \res(P \cup (\tent(I) \setminus X))$ and hence condition P5 for $(P_1,
X_1)$ follows from that for $(P, X)$.
Moreover, by condition P6 for $(P, X)$,  
each member of $\res_h(P \cup (\tent(I) \setminus X))$ precedes all
members of $X$ in the preference list of $h$ in $I$. 
Since $X_1 \subseteq X$ and $I_1$ is an extension of $I$, 
it follows that each member of $\res_h(P_1 \cup (\tent(I_1) \setminus X_1))$
precedes all members of $X_1$ in the preference list of $h$ in $I_1$: condition
P6 holds. We have confirmed that $(P_1, X_1)$ is indeed a prescription for
$I_1$.

We note that $\tgs(P_1, X_1) = \tgs(P, X) \setminus \rej(\tau_1)$ is 
non-empty under our current assumption. 
Therefore, we may apply the induction hypothesis to instance $I_1$ and
prescription $(P_1, X_1)$ for $I_1$ to obtain a simple and
hospital-complete extension $I'_1$ of $I_1$ such that $\prop(I'_1) \setminus
\prop(I_1) \subseteq P_1$, $\rej(I'_1) \setminus \rej(I_1) \subseteq X_1$, and
$\tgs(P_1, X_1) \subseteq \rej(I'_1)$. 
We have 
\begin{eqnarray*}
\prop(I'_1) \setminus \prop(I) & = &(\prop(I'_1)
\setminus \prop(I_1)) \cup (\prop(I_1) \setminus \prop(I))\\
& \subseteq & P_1 \cup Q \\
& = & P, 
\end{eqnarray*}
\begin{eqnarray*}
\rej(I'_1) \setminus \rej(I)& \subseteq &
(\rej(I'_1) \setminus \rej(I_1)) \cup (\rej(I_1) \setminus \rej(I))\\
& \subseteq & X_1 \cup \rej(\tau_1) \\
& = & X, 
\end{eqnarray*}
and
\begin{eqnarray*}
\tgs(P, X) & \subseteq & \tgs(P_1, X_1) \cup \rej(\tau_1) \\
&\subseteq & \rej(I'_1) \cup \rej(\tau_1) \\
& = & \rej(I'_1),
\end{eqnarray*}
since $\rej(\tau_1) \subseteq \rej(I_1) \subseteq \rej(I'_1)$.
Therefore, setting $J = I'_1$, the statement of the lemma holds.
This completes the induction step and hence the proof of the lemma.
\qed
\end{proof}

\begin{lemma}
\label{lem:to-presc}
Let $I$ be a resident-minimal and hospital-complete instance
and $(r_0, h_0)$ a match in $\tent(I)$ that is not finalizable in $I$.
Let $\sigma$ be a maximal $I$-feasible sequence and
$\sigma + \tau$ a shortest feasible extension of $\sigma$ that rejects
$(r_0, h_0)$.
Then, $(\prop(\tau), \rej(\tau))$ is a 
prescription for $I$ with $\tgs(\prop(\tau), \rej(\tau)) = \{(r_0,
h_0)\}$.
\end{lemma}
\begin{proof}
We set $P = \prop(\tau) \setminus \rej(\tau)$ and
$X = \rej(\tau) \setminus \prop(\tau)$. It will turn out
that $\prop(\tau) \cap \rej(\tau) = \emptyset$ and hence
$P = \prop(\tau)$ and $X = \rej(\tau)$.

We first confirm that $(P, X)$ is a prescription for $I$.
Since $\sigma + \tau$ is feasible, $P \subseteq \prop(\tau)$ is
disjoint from $\prop(\sigma) = \prop(I)$: condition P1 holds.  Since $P
\subseteq \tent(\sigma + \tau)$, for each $r \in R$, 
there is at most one $h$ such that $(r, h)\in P$: condition P2 holds. 
Since 
each match rejected by $\tau$ but not already in $\tent(\sigma)$ must be in
$\prop(\tau)$, we have $X \subseteq \tent(I)$:
condition P3 holds.
For condition P4, let $r \in \res\,P \cap \res\,\tent(I)$.
As $(r, h)$ for some $h$ is proposed in $\tau$, some match $(r, h') \in
\tent(I)$ must be rejected in $\tau$ and hence in 
$\rej(\tau) \setminus \prop(\tau) = X$.
Therefore, we have $\res\,P \cap \res\,\tent(I) \subseteq \res\,X$.

For conditions P5 and P6, fix $h \in H$.
Since $P \cup (\tent(I) \setminus X) = \tent(\sigma +
\tau)$, we have $|\res_h(P \cup (\tent(I) \setminus X))| \leq q_h$.
Moreover, if $\res_h\, X$ is non-empty, then $\tau$ rejects 
a match involving $h$ and therefore this
inequality is tight. Therefore, condition P5 holds.
As each member of $\res_h(\tent(\sigma + \tau))$ precedes all members of
$\res_h(\rej(\tau))$, condition P6 holds.

We have $(r_0, h_0) \in \rej(\tau)$ and, from the assumption that $\tau$ is
chosen to be the shortest, $r_0$ is not involved in any proposal in $\tau$.
Therefore $(r_0, h_0) \in \tgs(P, X)$.  A match $(r, h)$ in $\tgs(P, X)$  
distinct from $(r_0, h_0)$ would also contradict that assumption,
since the rejection of such $(r, h)$ may be removed from $\tau$ without
affecting the feasibility as $r$ is not involved in any proposal in $\tau$.
We conclude that $\tgs(P, X) = \{(r_0, h_0)\}$. 

By Lemma~\ref{lem:executable}, there is a simple extension $J$ of
$I$ such that $\prop(J) \setminus \prop(I) \subseteq P$, 
$\rej(J) \setminus \rej(I) \subseteq X$, and
$(r_0, h_0) \subseteq \rej(J)$. Let $\sigma + \tau'$ be a maximum
$J$-feasible extension of $\sigma$. Then, $\prop(\tau') = 
\prop(J) \setminus \prop(I) \subseteq P \subseteq \prop(\tau)$ and 
$\rej(\tau') = \rej(J) \setminus \rej(I) \subseteq X \subseteq \rej(\tau)$.
All of these inclusions must in fact be equalities, since otherwise
$\sigma + \tau'$ is a feasible extension of $\sigma$ rejecting $(r_0, h_0)$
that is shorter than $\sigma + \tau$, a contradiction.
Therefore, we have $\prop(\tau') = \prop(\tau) = P$ and
$\rej(\tau') = \rej(\tau) = X$, finishing the proof of the lemma. 
\qed 
\end{proof}

We have focused on those resident-minimal instances that are 
also hospital-complete. The following theorem, however, is
on general resident-minimal instances.

\begin{theorem}
\label{thm:resident-minimal}.
Let $I$ be a resident-minimal instance
and $(r_0, h_0)$ a match in $\tent(I)$.
Then, the following three conditions are equivalent.
\begin{description}
  \item[(1)] Match $(r_0, h_0)$ is not finalizable in $I$.
  \item[(2)] There is a resident-changeless and hospital-complete
  extension $I'$ of $I$ such that there is a 
  prescription $(P, Y)$ for $I'$ with $(r_0, h_0) \in \tgs(P, Y)$.
  \item[(3)] There is a simple extension of $I$ that rejects $(r_0,
  h_0)$.
\end{description}
\end{theorem}
\begin{proof}
(3) $\Rightarrow$ (1) is trivial. We show (1) $\Rightarrow$ (2) $\Rightarrow$
(3) below.

(1) $\Rightarrow$ (2): Suppose $(r_0, h_0)$ is not finalizable in $I$.
Let $J$ be an extension of $I$ that rejects $(r_0, h_0)$.  Let 
$J'$ be an arbitrary resident-changeless and hospital-complete extension of
$J$.  We let $I'$ be the resident-changeless and hospital-complete
extension of $I$ in which the preference list of each hospital is identical
to that in $J'$.
By Proposition~\ref{prop:hospital-completion}, 
$I'$ is resident-minimal. Since $J'$ is an extension of $I'$ and
rejects $(r_0, h_0)$, by Lemma~\ref{lem:to-presc}, there is
a prescription $(P, Y)$ for $I'$ such that 
$(r_0, h_0) \in \tgs(P, Y)$.

(2) $\Rightarrow$ (3): Let $I'$ and $(P, Y)$ be as in condition (2).
By Lemma~\ref{lem:executable}, there is a simple extension $J$ of
$I'$ such that $\tgs(P, Y) \subseteq \rej(J)$.
Since $J$ is a simple extension of $I$, we are done.
\qed 
\end{proof}

This theorem shows that, for resident-minimal instance $I$, a triple $(P, Y,
I')$, where $I'$ is a resident-changeless and hospital-complete extension of 
$I$ and $(P, Y)$ is a prescription for $I'$,
is a certificate that each match in $\tgs(P, Y)$ is not
finalizable in $I$.  We seek a more concise certificate and
generalize the notion of prescription to general resident-minimal instances.

Let $I$ be a resident-minimal instance.
A {\em prescription for $I$} is a pair $(P, X)$ of sets of matches
that satisfies conditions P1, P2, P3, P4, P5 in Proposition~\ref{prop:presc}
together with the following condition that replaces P6.
\begin{description}
\item[P6':] For each $h \in H$, the following holds.
Each member of $\res_h(P \cup ((\tent(I) \setminus \pend(I))\setminus X)))$ 
precedes all members of $\res_h(X \setminus \pend(I))$ in the
preference list of $h$ in $I$.
Moreover, if $\res_h(X \setminus \pend(I))$ is non-empty then 
$\res_h\,\pend(I) \subseteq \res_h\,X$.
\end{description}

The target set $\tgs(P, X)$ of prescription $(P, X)$ is defined
in the same manner as in the special case before:
$\tgs(P, X) = \{(r, h) \in X \mid r \not\in \res\,P\}$. 

Note that if $I$ is hospital-complete then $\pend(I)$ is empty and hence
condition P6' is equivalent to condition P6.

\begin{lemma}
\label{lem:presc-to-general}
Let $I$ be a resident-minimal instance, $I'$ a resident-changeless
and hospital-complete extension of $I$, and $(P, Y)$ a 
prescription for $I'$. Then, $(P, X)$, where $X = Y \cup (\rej(I') \setminus
\rej(I))$ is a prescription for $I$.
\end{lemma}
\begin{proof}
Conditions P1 and P2 do not depend on $X$ and therefore follow
from those conditions for prescription $(P, Y)$. Since $Y \subseteq \tent(I') \subseteq
\tent(I)$ and $\rej(I') \setminus \rej(I) \subseteq \tent(I)$, condition 
P3 that $X \subseteq \tent(I)$ holds.
For condition P4, let $r \in \res\,P \cap \res\,\tent(I)$.
If $r \in \res\,P \cap \res\,\tent(I')$ then $r \in Y$ by condition P4 for 
prescription $(P, Y)$. Otherwise, $r \in \res(\rej(I') \setminus \rej(I))
\subseteq \res\,X$. Therefore, condition P4 holds.
Condition P5 is equivalent to condition P5 for prescription $(P, Y)$, 
since $P \cup (\tent(I) \setminus X) = P \cup (\tent(I') \setminus Y)$.
For condition P6, fix $h \in H$.
In the preference list of $h$ in $I'$,
each member of $\res_h(P \cup (\tent(I) \setminus X))$
precedes all members of $\res_h\,Y$ 
by condition P6 for prescription $(P, Y)$, and obviously 
precedes all members of $\res_h(\rej(I') \setminus \rej(I))$. Therefore, 
it precedes all members of $\res_h\,X$. 
If $r \not\in \res\,\pend(I)$ then $r$ is already in the preference list
of $h$ in $I$. Therefore, 
each member of $\res_h(P \cup ((\tent(I) \setminus \pend(I))\setminus X)))$ 
precedes all members of $\res_h(X \setminus \pend(I))$ in the
preference list of $h$ in $I$. Moreover, suppose some $r \in \res_h(X \setminus
\pend(I))$ and some $r' \in \res_h\,(\pend(I) \setminus X)$. 
Then, since $r'$ is in $\res_h(P \cup (\tent(I) \setminus X))$ and
$r \in \res_h\, Y$, $r'$ must precede $r$ in the preference list of
$h$ in $I'$.  But this is impossible since $r$ is on the preference list
of $h$ in $I$ while $r'$ is not, a contradiction.  Therefore, 
if $\res_h(X \setminus \pend(I))$ is non-empty then $\res_h\,\pend(I) \subseteq
\res_h\,X$: condition P6 holds.
\qed  
\end{proof}

\begin{lemma}
\label{lem:presc-from-general}
Let $I$ be a resident-minimal instance
and $(P, X)$ is a prescription for $I$.
Then, there is some resident-changeless and hospital-complete extension $I'$ of
$I$ such that $(P, Y)$, where $Y = X \setminus \rej(I')$, is
a prescription for $I'$.
\end{lemma}
\begin{proof}
For each $h \in H$, arbitrarily complete the preference list of $h$ in $I$ so
that the residents in $\res_h(\pend(I) \cap X)$ get the lowest ranks in the
completed list.  Let the resulting instance be $I'$.
We confirm that $(P, Y)$ is a prescription for $I'$.
Conditions P1 and P2 do not depend on $P$ and therefore follow from
those conditions for prescription $(P, X)$. Since $X \subseteq
\tent(I)$ and $\tent(I') = \tent(I) \setminus \rej(I')$, 
condition P3 that $Y \subseteq \tent(I')$ holds.
For condition P4, let $r \in \res\,P \cap \res\,\tent(I')$.
Since $r \in \res\,P \cap \tent(I)$, we have $r \in \res\,X$ by
condition P4 for prescription $(P, X)$. Therefore, we have
$r \in \res\,X \cap \res\,\tent(I') = \res(X \setminus \rej(I')) = \res\,Y$, 
condition P4.
Condition P5 is equivalent to condition P5 for prescription $(P, X)$, 
since $P \cup (\tent(I) \setminus X) = P \cup (\tent(I') \setminus Y)$.

To show that P6 holds, let $r \in \res_h(P \cup (\tent(I')
\setminus Y))) = \res_h(P \cup (\tent(I) \setminus X))$.  
Suppose first that $r \not\in \res_h\,\pend(I)$.
Then, by condition P6' for $(P, X)$, $r$ precedes
all members of $\res_h(Y \setminus \pend(I)) \subseteq \res_h(X \setminus
\pend(I))$ in the preference list of $h$ in $I$ and hence in $I'$ as well.
Since $r$ precedes all members of $\res_h\,\pend(I)$
in the preference list of $h$ in $I'$ by the way $I'$ completes the preference
list of $h$, we conclude that $r$ precedes all members of $\res_h\,X$ in that
preference list.
Suppose next that $r \in \res_h\,\pend(I)$. 
Then, since $\res_h\,\pend(I)$, having $r$ as a member, is not contained in
$\res_h\,X$, $\res_h(X \setminus \pend(I))$ is empty, by condition P6' for $(P, X)$.
Therefore $r$ precedes all members in $\res_h\,X$ in the preference list of $h$
in $I'$, as those members are placed in the lowest positions. 
In either case, $r$ precedes all members of $\res_h\,Y \subseteq \res_h\,X$
in the preference list of $h$ in $I'$, that is, condition P6 holds
for $(P, Y)$.
\qed  
\end{proof}

Thus, a prescription for a general resident-minimal instance is
indeed a certificate for the negative answer to the finalizability
of a tentative match.

\begin{theorem}
\label{thm:general-presc}
Let $I$ be a resident-minimal instance and $(r_0, h_0)$ a
match in $\tent(I) \setminus \pend(I)$. Then,
there is a prescription $(P, X)$ for $I$ with
$(r_0, h_0) \in \tgs(P, X)$ if and only if
there is some resident-changeless and hospital-complete
extension $I'$ of $I$ and a prescription $(P, Y)$ for
$I'$ with $(r_0, h_0) \in (P, Y)$.
\end{theorem}
\begin{proof}
Suppose first that there is a prescription $(P, X)$ for $I$ with
$(r_0, h_0) \in \tgs(P, X)$. By Lemma~\ref{lem:presc-from-general},
there is a resident-changeless and hospital-complete
extension $I'$ of $I$ and a prescription $(P, Y)$ for
$I'$ such that $Y = X \setminus \rej(I')$.  
As $X \cap \rej(I') \subseteq \pend(I)$ and $(r_0, h_0) \not\in \pend(I)$,
$(r_0, h_0) \in \tgs(P, X)$ implies $(r_0, h_0) \in \tgs(P, Y)$.
For the converse, suppose that there is
a resident-changeless and hospital-complete
extension $I'$ of $I$ and a prescription $(P, Y)$ for $I$ with
$(r_0, h_0) \in \tgs(P, Y)$.
By Lemma~\ref{lem:presc-to-general}, $(P, X)$, where 
$X = Y \cup (\rej(I') \setminus \rej(I))$, is a prescription
for $I$. Since $\tgs(P, Y) \subseteq \tgs(P, X)$, we have
$(r_0, h_0) \in \tgs(P, X)$.
\qed 
\end{proof}

We close this subsection by sketching an integer program (IP)
for computing a prescription for a given resident-minimal
instance $I$ and a match $(r_0, h_0) \in \tent(I)$.
More precisely, the IP captures a triple $(P, X, Z)$, where
$(P, X)$ is a prescription for $I$ with $(r_0, h_0) \in \tgs(P, X)$
and $Z$ is a subset of $\pend(I)$ such that there is a resident-changeless and
hospital complete extension $J$ of $I$ with $Z = \rej(J) \setminus
\rej(I)$ and $(P, X \setminus Z)$ being a prescription for $J$.

We only describe the variables in the IP and their 
intended interpretations. The linear constraints are straightforward
to write down based on those interpretations. All the variables are
binary.  For each match $(r, h) \not\in \prop(I)$, we have
a variable $p_{r, h}$: $p_{r, h} = 1$ if and only if $(r, h) \in P$.
For each match $(r, h) \in \tent(I)$, we have a variable
$x_{r, h}$: $x_{r, h} = 1$ if and only if $(r, h) \in X$.
For each $h \in H$ and a subset $S$ of $\res_h\,\pend(I)$,
we have a variable $z_{h, S}$: $z_{h, S} = 1$ if and only if
$\res_h\,Z = S$. The objective function is the sum of $p_{r, h}$
over all $(r, h) \in (R \times H) \setminus \prop(I)$, which is minimized.
The optimal solution of this IP corresponds to a desired prescription
with the smallest cardinality of $P$.

\subsection{Polynomial time algorithm for the stable marriage case}
In this subsection, we show that 1-FTM-RM, the finalizability of
a tentative match for resident-minimal stable marriage instances,
is polynomial time solvable. 

Let $I$ be a resident-minimal stable marriage instance.
We define a bipartite digraph $G_I$ on vertex sets
$T = \tent(I)$ and $P = (R \times H) \setminus \prop(I)$ as follows.
Let $(r, h) \in T$ and $(r', h') \in P$.
There is an edge from $(r, h)$ to $(r', h')$ if and only if
$r = r'$.  There is an edge from $(r', h')$ to $(r, h)$
if and only if $h = h'$, both $r$ and $r'$ are on the
preference list of $h$ in $I$, and $r'$ precedes $r$
in that list. 

\begin{lemma}
\label{lem:rejection-chain}
Let $I$ be a resident-minimal stable marriage instance
and $(r_0, h_0)$ a match in $\tent(I)$.
Then, there is a simple extension of $I$ that rejects
$(r_0, h_0)$ if and only if there is a directed path in $G_I$
from some root (a vertex without incoming edges) of $G_I$
to $(r_0, h_0)$.
\end{lemma}
\begin{proof}
Suppose first that $I$ has a simple extension $I'$ that rejects $(r_0, h_0)$.
Let $\sigma$ be a maximal $I$-feasible sequence and $\sigma + \tau$
a maximal $I'$-feasible sequence.  We determine a sequence of
matches $(r_i, h_i)$, $i = 0, 1, \ldots$, so that the reversed sequence
$(r_j, h_j)$, $j = i, i - 1, \ldots, 0$, forms a directed path 
from $(r_i, h_i)$ to $(r_0, h_0)$ in $I$, for each $i$.
We maintain the invariant that
if $(r_i, h_i) \in T$ then $(r_i, h_i) \in \rej(\tau)$ and
if $(r_i, h_i) \in P$ then $(r_i, h_i) \in \prop(\tau)$.
We sart with the given match $(r_0, h_0)$.

Suppose $i \geq 0$ and match $(r_i, h_i)$ 
has been determined. If $(r_i, h_i)$ is a root of $G_I$ then we are
done as we have a desired path from $(r_i, h_i)$ to $(r_0, h_0)$. 
Suppose otherwise.  First suppose
that $(r_i, h_i) \in T$. If $(r_i, h_i) \in \pend(I)$ then
$r_i$ is not on the preference list of $h_i$ in $I$ and
hence there is no incoming edge to $(r_i, h_i)$ in $G_I$.
Since we are assuming that $(r_i, h_i)$ is not a root of $G_I$,
we conclude that $(r_i, h_i) \in \tent(I) \setminus \pend(I)$.
Due to the invariant, $(r_i, h_i)$ is in $\rej(\tau)$ and hence  
its rejection must be preceded in
$\tau$ by a proposal of some match $(r, h_i)$ in $\prop(\tau) \subseteq  P$
such that $r$ precedes $r_i$ in the preference list of $h_i$ and hence
there is an edge of $G_i$ from $(r, h_i)$ to $(r_i, h_i)$.
We let $(r_{i + 1}, h_{i + 1}) = (r, h_i)$.
Next suppose $(r_i, h_i) \in P$. Then, by the invariant
we have $(r_i, h_i) \in \prop(\tau)$.  If $r_i \not\in \res\,\tent(I)$ then
$(r_i, h_i)$ is a root of $G_I$ and we are done.
Otherwise, the proposal of $(r_i, h_i)$ must be preceded
in $\tau$ by the rejection of $(r_i, h)$ for some $h$.
We let $(r_{i + 1}, h_{i + 1}) = (r_i, h)$.

As the construction selects matches appearing in $\tau$
in the reversed order, it must eventually end at a root of $G_I$.

For the converse, suppose there is a directed path $p$ from
some root of $G_I$ to $(r_0, h_0)$.  Let $\tau_p$ be an
event sequence listing the matches in $p$ in the same order and
making each match in $P$ a proposal and each match in $T$ a rejection.
Extend $I$ by adding $h$ to the preference list of $r$, for
each $(r, h) \in \prop(\tau_p)$.  Furthermore, if the starting
vertex $(r^*, h^*)$ of $p$ is in $T$, which implies that
$(r^*, h^*) \in \pend(I)$, complete the preference list of
$h^*$ so that $r^*$ gets the lowest rank. Let $I'$ be
the resulting extension of $I$. Let $\sigma$ be 
a maximal $I$-feasible sequence.
If $(r^*, h^*) \in T$ then, as the quota of each hospital is one, 
$\sigma + (r^*, h^*)^-$ is $I'$-feasible. Otherwise, since
$(r^*, h^*) \in P$ and $r^* \not\in \res\,\tent(I)$, it follows that 
$\sigma + (r^*, h^*)^+$ is $I'$-feasible. 
By a straightforward induction, we may verify that
$\sigma + \tau_p$ is $I'$-feasible. As $\rej(\tau_p) \subseteq \tent(I)$,
$\sigma + \tau_p$ is a simple extension of $\sigma$ and hence
$I'$ is a simple extension of $I$ that rejects $(r_0, h_0)$.
\qed 
\end{proof}

The following theorem is immediate from
Theorem~\ref{thm:resident-minimal} and Lemma~\ref{lem:rejection-chain}.

\begin{theorem}
\label{thm:1-FTM-RM}
1-FTM-RM is solvable in polynomial time. 
\end{theorem}

\subsection{Hardness of FTM-RM}
\label{subsec:hardness-rm}
In this subsection, we show that 2-FTM-RM, and hence FTM-RM,
is coNP-complete. The reduction is from SAT through an intermediate
problem we call DIGRAPH-FIRING.

Let $G$ be a digraph and $\theta : V(G) \rightarrow N$ be
a {\em threshold} function which assigns a non-negative integer
$\theta(v)$ to each vertex $v$ of $G$. A {\em $\theta$-firing}
of $G$ is a subgraph $F$ of $G$ such that,
for each $v \in V(F)$, the in-degree of $v$ in $F$ is at least
$\theta(v)$ and the out-degree of $v$ in $F$ is at most 1.

\smallskip
\noindent\textbf{$k$-DIGRAPH-FIRING}
\newline\textbf{Instance:} A triple $(G, t, \theta)$,
where $G$ is a digraph, $t$ is a vertex of $G$, 
and $\theta$ is a threshold function on $V(G)$ such that
$\theta(v) \leq k$ for every $v \in V(G)$.
\newline\textbf{Question:} 
Does $G$ have a $\theta$-firing that contains
$t$?

\begin{lemma}
\label{lem:dag-firing}
2-DAG-FIRING is NP-Complete.
\end{lemma}
\begin{proof}
That 2-DAG-FIRING is in NP is trivial. We show its NP-hardness by
a reduction from SAT.  Let $S$ be a set of clauses, $X$ the
set of variables of $S$, and $C_1$, \ldots, $C_m$ the enumeration
of clauses in $S$.  Using Proposition~\ref{prop:1-2SAT}, we assume that
each variable in $X$ appears positively in exactly two clauses and 
negatively in exactly one clause.
For each $x \in X$, let $i_x^-$ denote the index of the clause
containing $x$ negatively and let $i_x^{1}$ and
$i_x^{2}$ denote the indices of clauses that contain $x$ positively.
We construct a DAG $G$ as follows.
$V(G)$ contains distinct vertices $a_i$ and $b_i$ for $1 \leq i \leq m$
and five distinct vertices $u_x$, $v_x$, $l^-_x$, $l^{1}_x$, and
$l^{2}_x$ for each $x \in X$. 
The edge set is defined by
\begin{eqnarray*}
E(G) & = & \{(a_i, b_i) \mid 1 \leq i \leq m\} 
 \cup \{(b_i, b_{i + 1}) \mid 1 \leq i < m\}  \cup \bigcup_{x \in X} E_x,  
\end{eqnarray*} 
where
\begin{eqnarray*}
E_x & = & \{(u_x, l^-_x), (v_x, l^-_x), (u_x, l^1_x), (v_x, l^2_x),
(l^-_x, a_{i^-_x}), (l^1_x, a_{i^1_x}), (l^2_x, a_{i^2_x})\}.
\end{eqnarray*} 
We set $t = b_m$. 
The threshold function $\theta$
is such that $\theta(v)$ is the in-degree of $v$ except 
that $\theta(a_i) = 1$ for $1 \leq i \leq m$. 
Since the only vertices with indegree possibly larger than two are $a_i$, $1
\leq i \leq m$, we have $\theta(v) \leq 2$ for every $v \in V(G)$.

See Fig.~\ref{fig:reduction1} for an example.

\begin{figure}[htb]
\begin{center}
\includegraphics[scale=0.5,bb=0 0 571 343]{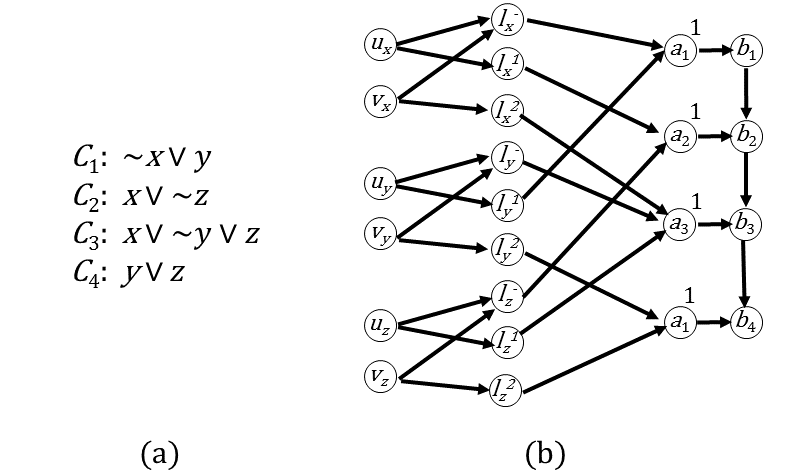}
\end{center}
(a) Clause set $S$\\
(b) DAG $G$; threshold is equal to the in-degree 
unless explicitly specified on the shoulder
\caption{Reduction from SAT to 2-DAG-FIRING}
\label{fig:reduction1}
\end{figure}

Observe that, for each $x \in X$, a    
$\theta$-firing cannot contain either $l^1_x$ or $l^2_x$ if 
it contains $l^-_x$ but can contain both
$l^1_x$ and $l^2_x$ simultaneously if it does not
contain $l^-_x$. Given this property of the ``variable gadgets''
in $G$,
it is straightforward to see that there is a mutual conversion 
between a satisfying
assignments of $S$ and a $\theta$-firing of $G$ containing $t$.
\qed
\end{proof}

\begin{theorem}
\label{thm:coNP2}
2-FTM-RM is coNP-complete.
\end{theorem}
\begin{proof}
That 2-FTM-RM is in coNP is trivial.
To show that it is coNP-hard, we 
we give a polynomial time reduction from $k$-DAG-FIRING to 
the complement of $k$-FTM-RM, for each positive integer $k$.
As 2-DAG-FIRING is NP-complete by Lemma~\ref{lem:dag-firing},
the theorem follows.

Fix $k$. Let $(G, t, \theta)$ be an instance of $k$-DAG-FIRING.
Without loss of generality, we assume that $t$ is a sink of $G$.
For each $v \in G$, let $\Nin(v)$ denote the set of
in-neighbors of $v$ in $G$ and let $V_0 = \{v \in V(G) \mid
\Nin(v) = \emptyset\}$ denote the set of roots of $G$.
We construct an instance 
$I = (R, H, \{q_h\}_{h \in H},$
$\{\pi_h\}_{h \in H}, 
\{\lambda_r\}_{r \in R})$ as follows.
For each $v \in V(G)$, we have a mutually distinct resident $r_v$
and we set $R = \{r_v \mid v \in V(G)\}$.
For each non-root vertex $v \in V(G) \setminus V_0$, we have a mutually distinct 
hospital $h_v$ and we set $H = \{h_v \mid v \in V(G) \setminus V_0\}$.
For each non-root vertex $v \in V(G) \setminus V_0$, 
we set $q_{h_v} = \theta(v)$.
For each non-root vertex $v \in V(G) \setminus V_0$, 
the preference list of $h_v$ lists
$r_u$, $u \in \Nin(v)$, in the first $|\Nin(v)|$ places
in an arbitrary order and then lists $r_v$ as its final element.
For each root $v \in V_0$, the preference list of
$r_v$ is empty (nothing disclosed). For each non-root vertex
$v \in V(G) \setminus V_0$, the preference list of $r_v$
consists of a single entry $h_v$ (only the top preference is disclosed). 
Finally, the match for which we ask the finalizability
is $(r_t, h_t)$. 
It is straightforward to verify that $I$ is resident-minimal and that 
and $(r_t, h_t) \in \tent(I)$; in fact we have $(r_v, h_v) \in \tent(I)$
for every $v \in V(G) \setminus V_0$. It is
also clear that the quota of each hospital in $I$ is $k$ or smaller.

See Figure~\ref{fig:reduction2} for an example.

\begin{figure}[htb]
\includegraphics[height=2in,bb=0 0 782 315]{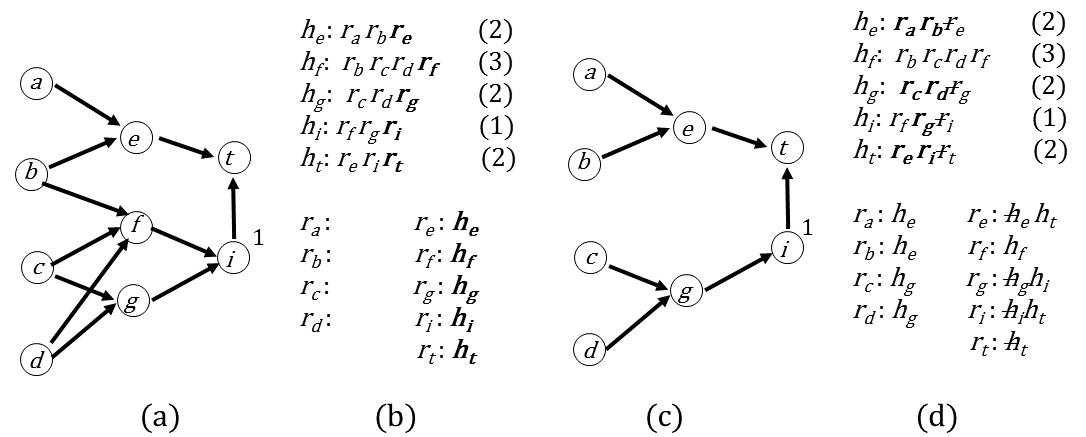}

\noindent(a) DAG $G$ for firing: the threshold equals the in-degree unless
explicitly specified on the shoulder\\
(b) The instance corresponding to $G$:
the parenthesized numbers are quotas;
tentative matches are shown in bold face from both sides\\
(c) A $\theta$-firing that contains $t$\\
(d) The instance extending the instance in (b) that corresponds to the firing in
(c); rejected matches are crossed out from both sides
\caption{Reduction from 3-DAG-FIRING to the complement of 3-FTM-RM}
\label{fig:reduction2}
\end{figure}

First suppose that $G$ has a $\theta$-firing $F$ that contains $t$.
We show that then $(r_t, h_t)$ is not finalizable in $I$.
Let $I'$ be obtained from $I$ by adding $h_v$
at the end of the preference list of $r_u$ in $I$, for each $(u, v) \in F$.
Let $v_1$, \ldots, $v_n = t$ be a topologically sorted enumeration of
$V(F)$.
We define $I'$-feasible sequence
$\sigma_i$, $0 \leq i \leq n$, inductively
as follows.  We will maintain the induction 
hypothesis that $\sigma_i$ is $I'$-feasible and
$\tent(\sigma_i) = \{(r_{v_j}, h_{v_j})
\mid v_j \in V(G) \setminus V_0 \mbox{\ and\ } j > i\}
\cup \{(r_{v_j}, h_{v_k}) \mid j \leq i \mbox{\ and\ }
(v_i, v_k) \in E(F)\}$.
Let $\sigma_0$ be an arbitrary
maximal $I$-feasible sequence. 
Since $I$-feasibility implies $I'$-feasibility and
$\tent(I) = \{(r_{v}, h_{v}) \mid v \in V(G) \setminus V_0\}$,
the induction hypothesis holds for the base case.
Suppose $i > 0$. 
If $v_i$ does not have
any outgoing edge in $F$ then set $\sigma_{i} = \sigma_{i - 1}$.
Suppose $v_i$ has an outgoing edge $(v_i, v_k)$ in $F$.
Because of the topological ordering, we have $k > i$.
If $v_i \in V_0$, then set
$\sigma_{i} = \sigma_{i - 1} + 
(r_{v_i}, h_{v_k})^+$. Since 
$r_{v_i} \not \in \res\,\tent(I)$ and
$h_{v_k}$ is ranked top in the preference list of $r_{v_i}$ in $I'$,
$\sigma_{i}$ is $I'$-feasible. The induction hypothesis is 
maintained since we have $(v_i, v_k) \in E(F)$ and
$(r_{v_i}, h_{v_k}) \in \tent(\sigma_i)$.
On the other hand, if $v_i \in V(F) \setminus V_0$ then 
set $\sigma_{i} = \sigma_{i - 1} + 
(r_{v_i}, h_{v_i})^- + (r_{v_i}, h_{v_k})^+$.
In this case, the in-degree of $v_i$ in $F$ is at least  
$\theta(v_i) = q_{h_{v_i}}$. For each in-neighbor $v_{j}$
of $v_i$ in $F$, its index $j < i$ and,
by the induction hypothesis, we have
$(r_{v_j}, h_{v_i}) \in \tent(\sigma_{i - 1})$.
Therefore, the preference list of $h_{v_i}$ in $I'$ has
at least $q_{h_{v_i}}$ residents in $\res_h \tent(\sigma_{i - 1})$ 
that precede $r_{v_i}$ and therefore 
$\sigma'_i = \sigma_{i - 1} + (r_{v_i}, h_{v_i})^-$ is
$I'$-feasible.  Moreover, since $r_{v_i}$ ranks $h_{v_k}$ immediately after
$h_{v_i}$, $\sigma_i = \sigma'_i + (r_{v_i}, h_{v_k})^+$ is $I'$-feasible.
We also have $\tent(\sigma_i) = \tent(\sigma_{i - 1}) \cup 
\{(r_{v_i}, h_{v_k})\} \setminus \{(r_{v_i}, h_{v_i})\}$
and therefore the induction hypothesis is maintained.
This construction leads to a $I'$-feasible sequence 
$\sigma_n$ that rejects $(r_t, r_t)$.
Therefore, $(r_t, h_t)$ is not finalizable in $I$.

For the converse, suppose $(r_t, h_t)$ is not finalizable in $I$.
Since $I$ is resident-minimal, by Theorem~\ref{thm:resident-minimal},
there is some simple extension $I'$ of $I$ that rejects $(r_t, h_t)$.
We assume without loss of generality that $I'$ is resident-minimal: take an
appropriate truncation if not. 
Let $U = V_0 \cup \{v \in V(G) \setminus V_0 \mid (r_v, h_v) \in
\rej(\tau)\}$ and let $F$ be the subgraph of $G$ induced by $U$.
We show that $F$ is a $\theta$-firing of $G$.
We first show that the out-degree of each vertex in $F$ is at most 1.
Let $(u, v)$ be an arbitrary edge of $F$.
By the definition of $G$, $r_u$ precedes $r_v$ in the preference
list of $h_v$ in $I$. Since $v \in U$, $(r_v, h_v)$ is rejected by $\tau$,
which implies that $\tau$ contains the proposal of $(r, h_v)$
for every $r$ that precedes $r_v$ in the preference list of $h_v$
(recall that there are exactly $q_{h_v}$ such residents $r$),
including $r_u$.  Therefore, the extension of the preference
list of $r_u$ from $I$ to $I'$ is by $h_v$.  
This show that, for each $u$, the vertex $v$ such that $(u, v)$ is an edge
of $F$ is unique if one exists: the out-degree of each vertex in $F$ is at
most 1.

We next show that the in-degree of each vertex $v$ is at least
$\theta(v)$. If $v \in V_0$, this is obvious
since $\theta(v) = 0$.  Suppose $v \in U \setminus V_0$.
Then, since $\tau$ rejects $(r_v, h_v)$,
this rejection event must be preceded in $\tau$ by
the proposal of $(r_u, h_v)$ for every resident $r_u$ in the set of
the $q_{h_v} = \theta(v)$ residents that precede $r_v$ in the preference list of $h_v$.
But for each such resident $r_u$, either $u$
is in $V_0$ (hence $r_u \not\in \res\;\tent(I)$)
or the rejection of $(r_u, h_u)$ precedes the proposal
of $(r_u, h_v)$ in $\tau$.  In either case,
we have $u \in U$.  Therefore, the in-degree of
$v$ is in $F$ is at least $\theta(v)$.
We conclude that $F$ is a $\theta$-firing of $G$.
Since $\tau$ rejects $(r_t, h_t)$,
we have $t \in U$.  This completes the proof
that if $(r_t, h_t)$ is not finalizable then 
then there is a $\theta$-firing of $G$ that contains $t$.
\qed 
\end{proof}

\section{A sufficient condition for finalizability}
\label{sec:safe}
In this section, we introduce a polynomial-time decidable sufficient condition
for a match to be finalizable in a given instance.
This condition turns out necessary for resident-minimal instances in the
stable marriage case, thus providing another proof that 
1-FTM-RM is polynomial time solvable (Theorem~\ref{thm:1-FTM-RM}).

Let $I$ be an instance and $M$ a subset of $\tent(I)$. 
We say that $r \in R$ is {\em relevant to} $h \in H$
with respect to $M$ if $r$ is matched in $M$ either to $h$ or
to no hospital.
We say that a match $(r, h)$ in $M$ is {\em endangered in $M$ with respect to
$I$} if it satisfies the following condition:
if $(r, h) \in \tent(I) \setminus \pend(I)$ then  
the preference list of $h$ in $I$ contains $q_h$ or more residents before $h$
that are relevant to $h$ with respect to $M$; if $(r, h) \in \pend(I)$ then the
number of residents relevant to $h$ with respect to $M$ is $q_h + 1$ or greater.
We denote by $\dang_I(M)$ the set of endangered matches with respect to $M$ in
$I$. 
We say that the set $M$ is {\em safe} with respect to $I$ if $\dang_I(M) =
\emptyset$. Observe that $\dang_I$ is monotone decreasing in the following
sense: if $M \subseteq M' \subseteq \tent(I)$ and $(r, h) \in M \setminus
\dang_I(M)$ then $(r, h) \not\in \dang_I(M')$.

See Table~\ref{tab1} for an example.

\begin{table}[htbp]
\caption{An example of a safe set}{%
\begin{center}
\begin{tabular}{cc}
\begin{minipage}{0.65\hsize}
Preference list $\pi_h$ of each $h \in H$\\
Matches in $\tent(I)$ are in bold face\\
Matches in $M$, a safe set, are parenthesized
\smallskip 
\begin{center}
\begin{tabular}{|c|c|c|c|c|c|c|c|c|c|}
\hline
$h$ & 1st & 2nd & 3rd & 4th & 5th & 6th & 7th & 8th & 9th \\ \hline
X& ({\bf a}) & i & e & ({\bf c}) & b & {\bf f} & d & - & -  \\ \hline
Y& i &({\bf g}) & a & ({\bf b}) & d & {\bf e} & c & - & -\\\hline 
Z& e & b & g & a & ({\bf i}) & d & - & - & - \\
\hline
\end{tabular}
\end{center}
\smallskip
Preference list of relevant residents for each $h$\\
Each member of $M$ is within
the quota of 3 in these lists
\smallskip
\begin{center}
\begin{tabular}{|c|c|c|c|c|c|c|c|c|c|}
\hline
$h$ & 1st & 2nd & 3rd & 4th & 5th & 6th & 7th & 8th & 9th \\ \hline
X& ({\bf a}) &  & e & ({\bf c}) &  & {\bf f} & d & - & -  \\ \hline
Y&  &({\bf g}) &  & ({\bf b}) & d & {\bf e} & c & - & -\\\hline 
Z& e &  &  &  & ({\bf i}) & d & - & - & - \\
\hline
\end{tabular}

\end{center}
\end{minipage}
\begin{minipage}{0.05\hsize}
\end{minipage}
\begin{minipage}{0.3\hsize}
Preference list of each $r \in R$; matches in $\tent(I)$ are in bold
face; rejected matches are in braces;
\smallskip
\begin{center}
\begin{tabular}{|c|c|c|c|}
\hline
$r$ & 1st & 2nd & 3rd \\ \hline
a& {\bf X}&-&- \\ \hline
b& {\bf Y}&X&- \\\hline
c& \{Y\}&{\bf X}&-  \\\hline
d& \{X\}&-&- \\\hline
e& {\bf Y}&X&- \\\hline
f& \{Y\}&{\bf X}&- \\\hline
g& \{X\}& {\bf Y}&- \\\hline
h& \{Y\}&-&- \\\hline
i& {\bf Z}&-&- \\
\hline
\end{tabular}
\end{center}
\end{minipage}
\end{tabular}
\end{center}
}
\label{tab1}
\end{table}

\begin{proposition}
\label{prop:safe}
Let $I$ be an instance and suppose $M \subseteq \tent(I)$ is safe
with respect to $I$. Then, every match in $M$ is finalizable in $I$.
\end{proposition}
\begin{proof}
Let $\sigma$ be an arbitrary feasible extension of the 
maximal $I$-feasible sequence.
Let $\tau$ be the maximal prefix of $\sigma$ that does not
contain the rejection of any member of $M$.  Since $M \subseteq \tent(I)$,
$\tau$ is an extension of the maximal $I$-feasible sequence.
Since $M \subseteq \tent(I)$ and $\tau$ does not reject any match in $M$, 
we have $M \subseteq \tent(\tau)$ and hence 
$\dang_I(\tent(\tau)) \subseteq \dang_I(M) = \emptyset$ by the monotonicity of $\dang_I$ observed above
and the assumption that $M$ is safe. This means that there is no match $(r, h)$
in $M$ such that $\tau + (r, h)^-$ is feasible, since if there is such a match 
then it would be endangered in $\tent(\tau)$ with respect to $I$.
Therefore, $\tau$ must be equal to $\sigma$ and therefore there is no 
extension of $I$ that rejects any match in $M$.
\qed 
\end{proof}

\begin{theorem}
\label{thm:safe-unique}
Let $I$ be an instance. Then,
the maximal safe set with respect to $I$ is unique and can be
identified in polynomial time.
\end{theorem}
\begin{proof}
Let $M_0 = \tent(I)$ and $M_i = M_{i - 1} \setminus \dang_I(M_{i - 1})$
for $i > 0$.  Let $m$ be the smallest $i$ such that $M_i = M_{i + 1}$.
Since $\dang_I(M_m) = \emptyset$, $M_m$ is safe. To show its maximality,
let $M$ be an arbitrary subset of $\tent(I)$ that is safe with respect to
$I$.  We show by induction on $i$ that $M \subseteq M_i$.
The base case $i = 0$ is trivial.  Suppose $i > 0$.
By the induction hypothesis, $M \subseteq M_{i - 1}$.
Let $(r, h)$ be a match in $\dang_I(M_{i - 1})$.
We cannot have $(r, h) \in M$, since if we did then, by
the monotonicity of $\dang_I$, $(r, h)$ would be endangered in $M$,
a contradiction to the assumption that $M$ is safe.
Therefore, $\dang_I(M_{i - 1}) \cap M = \emptyset$ and hence
$M \subseteq M_i$ holds.  Therefore we have $M \subseteq M_m$ and
hence $M_m$ is the unique maximal safe set.
\qed 
\end{proof}

In the stable marriage resident-minimal case, 
the above sufficient condition for
finalizability turns out necessary as well.
\begin{theorem}
\label{thm:necessary}
Let $I$ be a resident-minimal instance in the stable marriage case.
Then, each $(r, h) \in \tent(I)$ is finalizable only if
it is in the maximal safe set with respect to $I$. 
\end{theorem}
\begin{proof}
Fix $I$ and let $M_i$, $0 \leq i \leq m$, be as
defined in the proof of Theorem~\ref{thm:safe-unique}.
In particular, $M_m$ is the maximal safe set with respect to $I$.
Fix an arbitrary match $(r, h) \in \tent(I) \setminus M_m$.
We show that $(r, h)$ is not finalizable.

Let $i$ be the smallest integer such that $(r, h) \not\in M_i$.
Since $(r, h) \in M_0 = \tent(I)$, we have $i > 0$.
We construct a sequence of matches $(r, h) = (r_i, h_i)$,
$(r_{i - 1}, h_{i - 1})$, \ldots, $(r_0, h_0)$ such that
$(r_j, h_j) \in \dang_I(M_{j - 1})$ for each $j$, $1 \leq j \leq i$.

Let $0 < j < i$ and suppose $(r_k, h_k) \in \dang_I(M_{k - 1})$ 
for $i \geq k > j$ has been determined. As $(r_{j + 1}, h_{j + 1}) \in
\dang_I(M_j)$, we have some resident, say $r_j$, that is relevant
to $h_{j + 1}$ with respect to $M_j$ and precedes $r_{j + 1}$ 
in the preference list of $h_{j + 1}$ in $I$.  Observe here that
we cannot have $(r_{j + 1}, h_{j + 1}) \in \pend(I)$ since if we had then
$(r_{j + 1}, h_{j + 1})$ would be in $\dang_I(M_0)$ and hence
not in $\dang_I(M_j) \subseteq M_j$ as $j \geq 1$.
Now, $r_j$ is relevant to $h_{j + 1}$ with respect to $M_j$ but
not with respect to $M_{j - 1}$ since $(r_{j + 1}, h_{j + 1})$ is
not in $\dang_I(M_{j - 1})$. Therefore, there is some hospital, say $h_j$,
distinct from $h_{j + 1}$ such that $(r_j, h_j) \in M_{j - 1} \setminus M_j =
\dang_I(M_{j - 1})$. Thus, we have determined 
$(r_j, h_j) \in \dang_I(M_{j - 1})$ for the current $j$ and, inductively, 
for each $j$, $i \geq j \geq 1$.

As observed above, we have $(r_j, h_j) \in \tent(I) \setminus \pend(I)$
for $i \geq j \geq 1$ and, since $I$ is a stable marriage instance,
$h_j$ for $i \geq j \geq 1$ are pairwise distinct.@More straightforwardly, 
$r_j$ for $i \geq j \geq 1$ are pairwise distinct as there is at most
one match in $\tent(I)$ involving a particular resident. 

Let $\sigma$ be a maximal $I$-feasible sequence.
We now construct successive extensions $I_1$, \ldots, $I_m$ of $I$
and successive extensions $\sigma_1$, \ldots, $\sigma_m$ of $\sigma$,
based on the sequence of matches constructed above. 
Since $(r_1, h_1) \in \dang_I(M_0) = \dang_I(\tent(I))$, either $(r_1, h_1) \in
\pend(I)$ or $(r_1, h_1) \in \tent(I)\setminus \pend(I)$. If $(r_1, h_1) \in
\pend(I)$ then let $I_1$ be obtained from $I$ by completing the preference list of
$h_1$ so that $r_1$ is ranked lowest and
let $\sigma_1 = \sigma + (r_1, h_1)^-$. 
Suppose otherwise that $(r_1, h_1) \in \tent(I) \setminus \pend(I)$ then,
since $(r_1, h_1)$ is endangered in $\tent(I)$ with respect to $I$,
there must be some resident, say $r_0$, that precedes $r_1$ in the
preference list of $h_1$ in $I$ and is relevant to $h_1$ with respect
to $\tent(I)$. The latter condition implies that $(r_0, h_1) \not\in \tent(I)$.
Moreover, $(r_0, h_1)$ is not in $\prop(I)$ since if it were then
it would be impossible for $(r_1, h_1)$ to be in $\tent(I)$. 
Since $I$ is resident-minimal, it follows that 
$h_1$ is not in the preference list of $r_0$.
We let $I_1$ be obtained from $I$ by appending $h_1$ to the preference list of
$r_0$ and let $\sigma_1 = \sigma + (r_0, h_1)^+ + (r_1, h_1)^-$.
In either case, $\sigma_1$ is $I_1$-feasible.
In general, we maintain the invariant that $(r_j, h_j) \in \rej(\sigma_j)$
and $\sigma_j$ is $I_j$-feasible for
$1 \leq j \leq m$. Suppose $j > 1$ and $I_{j - 1}$ together with $\sigma_{j -
1}$ has been defined.
Let $I_j$ be obtained from $I_{j - 1}$ by appending $h_{j}$
to the preference list of $r_{j - 1}$ and let $\sigma_j = \sigma_{j - 1} + 
(r_{j - 1}, h_j)^+ + (r_j, h_j)^-$. 
As $(r_{j - 1}, h_{j-1}) \in \rej(\sigma_{j - 1})$, $\sigma_{j - 1} + 
(r_{j - 1}, h_j)^+$ is $I_j$-feasible.  Moreover, since $r_{j - 1}$ precedes
$r_j$ in the preference list of $h_j$ by constructin, $\sigma_j$ is
$I_j$-feasible.
We conclude that $(r, h) = (r_m, h_m)$ is not finalizable in $I$ since
the extension $I_m$ of $I$ rejects $(r_m, h_m)$.
\qed
\end{proof}

It follows as a corollary to Proposition~\ref{prop:safe}, 
Theorem~\ref{thm:safe-unique}, and Theorem~\ref{thm:necessary}
that 1-FTM-RM is polynomial time solvable giving another proof of  
Theorem~\ref{thm:1-FTM-RM}.

%

\section{Student-supervisor assignment: simulations}
\label{sec:simulations}
In this section, we present some simulation results on the
student-supervisor assignment procedure mentioned in the introduction.
The purpose is to demonstrate that there are realistic markets
in which multi-round matching procedures based on FTM can be effective.
Since the statics from the real market are not publicly disclosable, 
we resort to simulations.

\subsection*{Real market}
We first describe the real student-supervisor market in the author's
department. Every student in the final year of undergraduate study takes a 
full year project course as a part of the requirement for graduation. 
Every faculty member in the department supervises a project course. 
Supervisors have quotas as even as possible that sum up to 
the total number of students.
The assignment procedure takes place in the following steps.
\begin{enumerate}
  \item Students visit supervisors' labs to see the research activity there and get interviews if interested.
  \item Each supervisor submits a rank list of students to the central system. This list must be complete.
  The grade point information is provided to the supervisors. Typically, supervisors make their rank list
  based on the grade points and the score from the interviews.
  \item The rank lists of the supervisors are not public but are partially disclosed
  in the following manner: if the rank of student $s$ in the list of supervisor $p$
  is within the quota of $p$, then $s$ is notified of this fact.
  \item The first round of matching: each student submits a rank list of length up to 3 and
  the deferred acceptance algorithm is executed. Among the resulting tentative matches,
  those found finalizable by the sufficient condition in Section~\ref{sec:safe}
  are finalized. Both the student and the supervisor of each finalized
  match are notified.
  \item Each student $s$ without a finalized match is informed of the following:
  (1) the list of unfilled supervisors (supervisors for which the number
  of finalized matches is strictly smaller than their quota) and (2) the list of 
  unfilled supervisors $p$ such that
  the rank of $s$ in the rank list of $p$, after removing students who are finalized
  to supervisors other than $p$, is within the quota of $p$.  Note that,
  in the circumstances in (2), $s$ must be matched to $p$ in any stable matching
  provided that $s$ ranks $p$ the highest among all supervisors except
  those that rejected $s$ in the first round.
  \item The second round of matching: each student without a finalized match
  submit a complete rank list of unfilled supervisors. This rank list
  must be consistent with the rank list in the first round, in that they 
  agree on the ordering of common entries.
  Then the deferred  
  acceptance algorithm is executed to complete the assignment.
\end{enumerate}

The final outcome of the two rounds of matching is stable, assuming that
the rank lists of students in both rounds are consistent with their true
preferences.  This assumption might be disputable because of the
partial disclosures of the supervisors' rank lists before each round, described above.
There is no strategic reason for students to change their preference orders but
there may be psychological factors.
Though these disclosures are introduced for good reasons, we do not include this
ingredient in our simulation, partly because we want to avoid such disputes and
partly because it is difficult to model the influence of such disclosures on the
preferences of students.

\subsection*{Model}
We have a set $S$ of students and a set $P$ of supervisors.
To model the diversity of interests of students and of attractiveness of 
the supervisors,
we have a set $T$ of topics.  
Besides $|S|$, $|P|$, $|T|$ we have parameters $k$, $\sigma_1$,
$\sigma_2$, and $\sigma_3$ to be used below.

We have the following random variables, which are mutually 
independent except for the relationships explicitly described.
Each $s \in S$ has a score $g_s$ which has a normal distribution
with mean $0.5$ and standard deviation $\sigma_1$. 
Each student $s \in S$ has an
interest value $i_{s, t}$ on each topic $t \in T$ and each 
supervisor $p$ has attractiveness $a_{p, t}$ on each topic $t \in T$.
For each student $s$, the total interest $\sum{t \in T} i_{s, t}$
is fixed to 1 and the relative magnitude of $i_{s, t}$, $t \in T$,
is proportional to a random variable with mean 0.5 and
standard deviation $\sigma_3$.
For each supervisor $p$, the total attractiveness $\sum{t \in T} a_{p,
t}$ has a normal distribution with mean 1 and standard
deviation $\sigma_2$, truncated to fit in the interval $[0.5, 1.5]$.
Given this total attractiveness, the relative magnitude of $a_{p, t}$, $t \in
T$, is proportional to a random variable with mean 0.5 and
standard deviation $\sigma_3$.

Based on these random variables, the rank lists of
students and supervisors are determined a follows.
For each pair of student $s$ and supervisor $p$, 
let $\attraction(s, p) = \sigma_{t \in T} i_{s, t} a_{p, t}$ denote
the attraction between $s$ and $p$, which is the inner product
between the interest vector $i_s$ of $s$ and the attractiveness
vector $a_p$ of $p$.

Each student $s$ uses $\attraction(s, p)$ as a score to rank $p$ in the list.  
The rank lists of supervisors are based on the grade scores of students
and the results of interviews described as follows.
Each student $s$ has interviews with top $k$ supervisors
in the rank list of $s$. The score supervisor $p$ uses to rank student
$s$ is the grade score $g_s$ if $p$ does not interview $s$.
If $p$ does interview $s$, then the score is modified to reflect
the chemistry catalyzed by the interview.  We 
use a simplest model that the score in this case 
is $g_s$ + $\attraction(s, p)$. 
 
\subsection*{Procedure}
The procedure has a parameter $r$, a positive integer.
In the first round, the top $r$ of the rank list of
each student are submitted and the deferred acceptance
algorithm is executed on this truncated instance.
Among the resulting tentative matches, those found
finalizable by the the sufficient condition in Section~\ref{sec:safe} are
finalized. We count the number of tentative matches, 
the number of finalized matches, and the number of supervisors
that are completely filled by the finalized matches in the
first round.  

The second round could be executed using the complete
rank lists but, in this study, we are not interested in 
the final outcomes.
   
\subsection*{Simulation results}
In our simulation, we set  
$|S| = 100$ and $|P| = 10$, round numbers which are close to
the real numbers in the author's department. 
We also fix the following parameters: $|T| = 4$,
$k = 5$, $r = 3$, and $\sigma_1 = \sigma_2 = 0.1$.
We try several values of parameter $\sigma_3$, which 
controls the degree of diversity of the interests
of the students and of the attractiveness of supervisors.

We have run the simulation 100 times for each value of
$\sigma_3$ and recorded the average, minimum, and maximum
values of each quantity measured.
Table~\ref{tab:simulations} shows the results.

\begin{table}[htbp]
  \begin{center} 
    \begin{tabular}{|c|c|c|c|c|c|c|c|c|c|c|c|c|c|c|c|c|} 
    \hline
	$\sigma_3$ & \multicolumn{3}{|c|}{tentative matches}&
	\multicolumn{3}{|c|}{finalized matches} &
	\multicolumn{3}{|c|}{finalized$/$tentative} &
	\multicolumn{3}{|c|}{filled supervisors}\\
	\hline
	& avg. & min & max &  avg. & min & max 
	& avg. & min & max &  avg. & min & max \\
	\hline 
    0.1&40.16&30&55&35.72&30&50&0.89&0.71&1.0&3.7&3.0&5.0\\
    0.3&73.57&51&91&67.49&44&88&0.91&0.78&1.0&6.29&4.0&8.0\\
    0.5&82.02&64&94&77.37&58&93&0.94&0.82&1.0&7.08&6.0&9.0\\
    0.7&82.81&69&96&78.68&62&91&0.94&0.81&1.0&7.2&6.0&9.0\\	
	\hline
    \end{tabular}
    \caption{Simulation results}
	\label{tab:simulations}
  \end{center}
\end{table}

As $\sigma_3$ increases, both the number of tentative matches and
the number of finalized matches tend to increase, except that
the results for $\sigma_3 = 0.5$ and $\sigma_3 = 0.7$ do not show
a significant difference. This tendency is plausible, since the diversity
of interests and attractiveness would result in the diversity of
preferences. 

It might be rather surprising that the ratio of the number of finalized matches
over the number of tentative matches is consistently high:
on average, it is around or above 90\% for all values of $\sigma_3$.

With high diversity of interests and attractiveness ($\sigma_3 = 0.5$, for
example), on average, about 77 students out of 100 are finalized after
the first round and about 7 supervisors out of 10 are filled with finalized
matches.

These numbers are fairly close to those from the real supervisor assignment
results in the author's department. Thus, the savings in the evaluation efforts of the students
are enormous. The students finalized in the first round do not need to extend
their list beyond the top 3 supervisors and those who are not finalized may
concentrate on the small number of unfilled supervisors in the second
round.

We do not claim that our model captures the underlying mechanism of 
the real market well. 
In particular, the
model for the effect of interviews is too simplistic. Nonetheless,
the simulations do demonstrate that multi-round stable matching procedure
based on FTM can be effective for markets where the preferences of
participants are diverse and some prematch process, such as interviews,
helps nurturing ties between some pairs through which
each side of a pair ranks the other high.

\section{Future work}
Though the sufficient condition for finalizability given
in Section~\ref{sec:safe} is useful in matching procedures
for markets as studied in Section~\ref{sec:simulations}, 
exact determination of finalizability, if can be done with
a reasonable amount of computation, would further enhance
the merit of the multi-round approach.
The characterization of negative instances for FTM-RM given
in Section~\ref{sec:r-minimal} would be indispensable in 
developing practical algorithms for exact finalizability.

Applicability of the approach to larger markets such as
NRMP is an interesting and challenging topic.  

\subsection*{Acknowledgments}
The author would like to thank his department colleagues,
especially Toshiyuki Tsutsumi, for discussions that were helpful in
forming the idea of incremental submissions.
He is also thankful to Shuichi Miyazaki 
for helpful comments on an earlier manuscript and to anonymous reviewers of SAGT
2015 conference whose constructive comments on a submission 
describing this work in a preliminary stage have lead to the
improved results and presentations in the current paper.
He also thanks SAGT 2016 reviewers for comments that were helpful in improving the
presentation.

\bibliographystyle{plain}
\bibliography{ec}

\end{document}